\pdfoutput=1

\documentclass[format=acmsmall,prodmode,acmtecs,dvipsnames,preprint]{acmart} 

\acmVolume{-}
\acmNumber{-}
\acmArticle{-}
\acmYear{2018}
\acmMonth{0}

\usepackage[english]{babel}
\usepackage{verbatim}
\usepackage{suffix}
\usepackage{xcolor}
\usepackage{mathtools}
\usepackage{braket}
\usepackage{amssymb}
\usepackage{MnSymbol} 
\usepackage{bbm}
\usepackage{url}
\usepackage{graphicx}
\usepackage{thm-restate}
\usepackage{soul}
\usepackage{algpseudocode}
\usepackage{algorithm}
\usepackage{mdwlist}
\usepackage{multirow}
\usepackage{balance}
\usepackage{xspace}
\usepackage{enumitem}
\usepackage{hyperref}
\usepackage[capitalize]{cleveref}
\usepackage{outlines}
\usepackage{array}
\usepackage{tabularx}
\usepackage{dcolumn}
\usepackage{subcaption}
  \newcolumntype{d}[1]{D{.}{.}{#1|}}
\usepackage{booktabs}
\usepackage{todonotes}

\title{\sys: A Framework for Defining  Differentially-Private Computations}

\author{Dan Zhang}
\affiliation{
  \department{College of Information and Computer Sciences}
  \institution{University of Massachusetts, Amherst}
}
\email{dzhang@cs.umass.edu}

\author{Ryan McKenna}
\affiliation{
  \department{College of Information and Computer Sciences}
  \institution{University of Massachusetts, Amherst}
}
\email{rmckenna@cs.umass.edu}

\author{Ios Kotsogiannis}
\affiliation{
  \department{Department  of Computer Science}
  \institution{Duke University}
}
\email{iosk@cs.duke.edu}

\author{George Bissias}
\affiliation{
  \department{College of Information and Computer Sciences}
  \institution{University of Massachusetts, Amherst}
}
\email{gbiss@cs.umass.edu}

\author{Michael Hay}
\orcid{0000-0001-9085-893X}
\affiliation{
  \department{Computer Science Department}
  \institution{Colgate University}
}
\email{mhay@colgate.edu}

\author{Ashwin Mach\-anav\-ajjhala}
\affiliation{
  \department{Department  of Computer Science}
  \institution{Duke University}
}
\email{ashwin@cs.duke.edu}

\author{Gerome Miklau}
\affiliation{
  \department{College of Information and Computer Sciences}
  \institution{University of Massachusetts, Amherst}
}
\email{miklau@cs.umass.edu}


\renewcommand{\shortauthors}{D. Zhang, R. McKenna, I. Kotsogiannis, G. Bissias, M. Hay, A. Machanavajjhala, G. Miklau}

\setcopyright{none}
\crefname{section}{Sec.}{Secs.}

\DeclareMathOperator*{\argmin}{arg\,min}

\newcommand{\privOp}{{\color{BrickRed} Private}\xspace}
\newcommand{\privPubOp}{{\color{Orange} Private$\rightarrow$Public}\xspace}
\newcommand{\pubOp}{{\color{OliveGreen} Public}\xspace}

\newcommand{\squishlist}{
	\begin{list}{$\bullet$}
		{
			\setlength{\itemsep}{0pt}
			\setlength{\parsep}{3pt}
			\setlength{\topsep}{3pt}
			\setlength{\partopsep}{0pt}
			\setlength{\leftmargin}{1.5em}
			\setlength{\labelwidth}{1em}
			\setlength{\labelsep}{0.5em} }
		}

\newcommand{\squishend}{
\end{list}  }

\newcommand{\op}[1]{{\sc #1}\xspace}

\newcommand{\eat}[1]{}

\renewcommand{\paragraph}[1]{\vskip 3pt \noindent {\em #1} }
\renewcommand{\hl}[1]{#1}

\newcommand{\ahpOp}[0]{{\algoname{\hl{AHPpartition}}}\xspace}
\newcommand{\reduceOp}[0]{\op{\hl{V-ReduceByPartition}}\xspace}
\newcommand{\partitionOp}[0]{\op{\hl{V-SplitByPartition}}\xspace}
\newcommand{\tablePartitionOp}[0]{\op{\hl{SplitByPartition}}\xspace}

\newcommand{\planref}[1]{\textnormal{\textsc{Plan \##1}}}

\def\planidentity{{1}\xspace}

\def\planmwem{{7}\xspace}
\def\planahp{{8}\xspace}
\def\plandawa{{9}\xspace}

\def\plandawastripe{{14}\xspace}
\def\planhbstripe{{15}\xspace}
\def\planhbstripekron{{16}\xspace}
\def\planprivbayesls{{17}\xspace}
\def\planmwemb{{18}\xspace}
\def\planmwemc{{19}\xspace}
\def\planmwemd{{20}\xspace}

\newcommand{\pk}{protected kernel\xspace}

\def\sys{\textsc{$\epsilon$ktelo}\xspace}

\newcommand{\algoname}[1]{\textnormal{\textsc{#1}}}

\newcommand{\norm}[1]{\left\lVert#1\right\rVert}

\def\QH{\mathcal{Q}}
\def\budget{\epsilon_{tot}}
\def\C{\mathbb{C}}
\def\Czero{\mathbb{C}_0}

\def\client{S_{client}}
\def\api{S_{kernel}}
\def\init{\textrm{\bf Init}}
\def\Q{\mathbf{M}}   
\def\W{\mathbf{W}}   
\def\P{\mathbf{P}}   
\def\qsize{m}
\def\wsize{m}
\def\psize{p}
\def\D{\mathbf{D}} 
\def\A{\mathbf{A}}
\def\B{\mathbf{B}}

\def\db{T}  

\def\algG{\mathcal{A}}  

\newcommand{\vect}[1]{\mathbf{#1}}
\def\x{\vect{x}}
\def\xhat{\vect{\hat{x}}}
\def\y{\vect{y}}
\def\q{\vect{q}}
\def\m{\vect{m}}
\def\u{\vect{u}}
\def\v{\vect{v}}
\def\z{\vect{z}}

\begin{document}

\begin{abstract}

The adoption of differential privacy is growing but the complexity of designing private, efficient and accurate algorithms is still high.  We propose a novel programming framework and system, \sys, for implementing both existing and new privacy algorithms.  For the task of answering linear counting queries, we show that nearly all existing algorithms can be composed from operators, each conforming to one of a small number of operator classes.  While past programming frameworks have helped to ensure the privacy of programs, the novelty of our framework is its significant support for authoring accurate and efficient (as well as private) programs.

After describing the design and architecture of the \sys system, we show that \sys is expressive, allows for safer implementations through code reuse, and that it allows both privacy novices and experts to easily design algorithms.  We provide a number of novel implementation techniques to support the generality and scalability of \sys operators.  These include methods to automatically compute lossless reductions of the data representation, implicit matrices that avoid materialized state but still support computations, and iterative inference implementations which generalize techniques from the privacy literature.

We demonstrate the utility of \sys by designing several new state-of-the-art algorithms, most of which result from simple re-combinations of operators defined in the framework.  We study the accuracy and scalability of \sys plans in a thorough empirical evaluation.

\end{abstract}


\maketitle



\section{Introduction} \label{sec:intro}

As the collection of personal data has increased, many institutions face an urgent need for reliable privacy protection mechanisms. They must balance the need to protect individuals with demands to use collected data for new applications, to model their users' behavior, or share  data with external partners. Differential privacy \cite{dwork2006calibrating,Dwork14Algorithmic} is a rigorous privacy definition that offers a persuasive assurance to individuals, provable guarantees, and the ability to analyze the impact of combined releases of data. Informally, an algorithm satisfies differential privacy if its output does not change too much when any one record in the input database is added or removed.

The research community has actively investigated differential privacy and algorithms are known for a variety of tasks ranging from data exploration to query answering to machine learning. However, the adoption of differentially private techniques in real-world applications remains rare. This is because implementing programs that provably satisfy privacy and ensure sufficient utility for a given task is still extremely challenging for non-experts in differential privacy. In fact, the few real world deployments of differential privacy -- like  OnTheMap~\cite{onthemap,haney17:census} (a U.S. Census Bureau data product), RAPPOR \cite{Erlingsson14Rappor:} (a Google Chrome extension), and Apple's private collection of emoji's and HealthKit data -- have required teams of privacy experts to ensure that implementations meet the privacy standard and that they deliver acceptable utility. There are at least three important challenges in implementing and deploying differentially private algorithms.

The first and foremost challenge is the difficulty of designing utility-optimal algorithms: i.e., algorithms that can extract the maximal accuracy given a fixed ``privacy budget.'' 
While there are a number of general-purpose differentially private algorithms, such as the Laplace Mechanism~\cite{dwork2006calibrating},
they typically offer suboptimal accuracy if applied directly.  A carefully designed algorithm can improve on general-purpose methods by an order of magnitude or more---without weakening privacy: that is, accuracy is improved by careful engineering and sophisticated algorithm design.

One might hope for a single dominant algorithm for each task, but a recent empirical study~\cite{Hay16Principled} showed that the accuracy of existing algorithms is complex: no single algorithm delivers the best accuracy across the range of settings in which it may be deployed.  The choice of the best algorithm may depend on the particular task, the available privacy budget, and properties of the input data.
Therefore, to achieve state-of-the-art accuracy, a practitioner currently has to make a host of complex algorithm choices, which may include choosing a low-level representation for the input data, translating their queries into that representation, choosing among available algorithms, and setting parameters.  The best choices will vary for different input data and different analysis tasks.

The second challenge is that the tasks in which practitioners are interested are diverse and may differ from those considered in the literature. Hence, existing algorithms need to be adapted to new application settings, a non-trivial task. For instance, techniques used by modern privacy algorithms include optimizing error over multiple queries by identifying common sub-expressions, obtaining noisy counts from the data at different resolutions, and using complex inference techniques to reconstruct answers to target queries from noisy, inconsistent and incomplete measurement queries. But different algorithms use different specialized operators for these sub-tasks, and it can be challenging to adapt them to new situations. Thus, designing utility-optimal algorithms requires significant expertise in a complex and rapidly-evolving research literature.

A third equally important challenge is that correctly implementing differentially private algorithms can be difficult. There are known examples of algorithm pseudocode in research papers not satisfying differential privacy as claimed. For instance, Zhang et al \cite{zhang16privtree} showed that many variants of a primitive called the sparse vector technique do not in fact meet their claims of differential privacy. Differential privacy can also be broken through incorrect implementations of valid algorithms. For example, Mironov \cite{mironov12laplace} showed that standard implementations of basic algorithms like the Laplace Mechanism \cite{dwork2006calibrating} can violate differential privacy because of their use of floating point arithmetic. Privacy-oriented programming frameworks such as PINQ \cite{mcsherry2009pinq,Ebadi17Featherweight,Proserpio14Calibrating}, Fuzz \cite{gaboardi:popl13}, PrivInfer \cite{privinfer:ccs16} and LightDP \cite{zhang:popl17} help implement programs whose privacy can be verified with relatively little human intervention. While they help to ensure the privacy criterion is met, they may impose their own restrictions and offer little or no support for designing utility-optimal programs. In fact, in PINQ \cite{mcsherry2009pinq}, some state-of-the-art algorithms involving inference and domain reduction cannot be implemented.

To address the aforementioned challenges, we have developed \sys, a programming framework and system that aids programmers in developing differentially private programs with high utility. \sys programs can be used to solve a core class of statistical tasks that involve answering counting queries over a table of arbitrary dimension (described in Sec. \ref{sec:background}). Tasks supported by \sys include releasing contingency tables, multi-dim\-ensional histograms, answering OLAP and range queries, and implementing private machine learning algorithms.  \sys is an open-source system under active development\footnote{Available at \url{https://github.com/ektelo/ektelo}}.

This paper makes the following contributions.

First, we recognize that, for the tasks we consider, virtually all algorithms in the research literature can be described as combinations of a small number of operators that perform basic functions. Our first contribution is to abstract and unify key subroutines into a small set of operator classes in \sys -- tranformations, query selection, \hl{partition selection}, measurement and inference. Different algorithms differ in (i) the sequence in which these operations are performed on the data, and (ii) the specific implementation of operations from these classes. In our system, differentially private programs are described as \emph{plans} over a high level library of \emph{operator} implementations supported by \sys. Plans described in \sys are expressive enough to reimplement all state-of-the-art algorithms from DPBench~\cite{Hay16Principled}.

Second, if operator implementations are vetted and shown to satisfy differential privacy, then plans implemented in \sys come with a proof of privacy. This proof requires a non-trivial extension of a formal analysis of a past framework~\cite{Ebadi17Featherweight}. This relieves the algorithm designer of the burden of proving their programs are private. By isolating privacy critical functions in operators, \sys reduces the amount of code that needs to be verified for privacy. In future work, we hope to implement operators in \sys using programming frameworks like LightDP to eliminate this burden too.


Third, we describe a number of novel implementation techniques which support the generality and efficiency of \sys.  We design sophisticated matrix support into \sys, which allows plan authors to represent and operate on matrix objects that would be infeasible to represent otherwise;  We describe a general-purpose, efficient and scalable inference engine that subsumes customized inference subroutines from the literature; and we describe a new dimensionality reduction operator that is applicable to plans that answer a workload of linear counting queries, and can reduce error by at most $3\times$ and runtime at most $5\times$.  These systems innovations are central to the goals of \sys because they give plan authors more freedom to describe plans without being limited by efficiency concerns or bound by the need to design custom inference techniques.

The operator-based approach to implementing differentially private programs has the following benefits:
\squishlist
	\item \emph{Modularity:} \sys enables code reuse since the same operator can be used in multiple algorithms.  This helps safety, as there is less code to verify the correctness of an implementation, and it amplifies innovation, as any improvement to an operator is inherited by all plans containing it.
	\item \emph{Transparency:} By expressing algorithms as plans with operators from operator classes, differences/similarities of competing algorithms can be discerned. Moreover, algorithm modifications easier to explore. Further, it is possible to identify general rules for restructuring plans (like  heuristics in query optimizers).
	\item \emph{Flexibility:} Practitioners can now use existing operators from different algorithms and recombine them in arbitrary ways -- allowing them to invent new algorithms that borrow ideas from the state-of-art -- without the need for a custom privacy analysis.
\squishend

We demonstrate the benefits of \sys by re-implementing a wide range of algorithms from the literature, using \sys to design original algorithms with improved error and efficiency, and by using \sys to address two case studies. For a use-case of releasing Census data tabulations, we define a new algorithm that offers a $10\times$ improvement over the best competitor from the literature. For building a private classifier, we used \sys to design algorithms that beat all available baselines.


\paragraph{Organization}
We provide an overview of \sys and highlight its design in the next section. After providing background in \cref{sec:background}, we describe the execution framework and methods for privacy enforcement \cref{sec:framework}.  The operator classes are described in \cref{sec:operators} and we show the expressiveness of \sys plans in  \cref{sec:plans} by re-implementing existing algorithms.  Efficient matrix support is described in \cref{sec:implementation}, including improvements to inference.  \cref{sec:wbdr} describes our method for selecting and optimal partition given a workload. We then put \sys into action by designing new algorithms for cases studies~\cref{sec:case_studies}, followed by a thorough experimental evaluation in \cref{sec:experiments}. We discuss related work and conclude in \cref{sec:related,sec:conc}.

\section{Overview and Design Principles} \label{sec:advantages}
In this section we provide an overview of \sys by presenting an example algorithm written in the framework.  Then we discuss the principles guiding the design of \sys.

\subsection{An example plan: CDF estimation}

In \sys, differentially private algorithms are described using \emph{plans} composed over a rich library of \emph{operators}. Most of the plans described in this paper are linear sequences of operators, but \sys also supports plans with iteration, recursion, and branching. Operators supported by \sys perform a well defined task and typically capture a key algorithm design idea from the state-of-the-art.  Each operator belongs to one of five \emph{operator classes} based on its input-output specification. These are: (a) transformation, (b)  query, (c) inference, (d) query selection, and (e) \hl{partition selection}.  Operators are fully described in Sec. \ref{sec:operators} and listed in Fig. \ref{fig:operator_index}.
%
\newcommand{\COMMENT}[2][.5\linewidth]{%
  \leavevmode\hfill\makebox[#1][l]{$\triangleright$~#2}}

\renewcommand{\COMMENT}[2][.5\linewidth]{\Comment #2}

\setlength{\textfloatsep}{4pt}

\begin{algorithm}[t]
	\caption{\sys CDF Estimator}

	\label{alg:cdf}
	\begin{algorithmic}[1]
		\State $D \gets$ \algoname{Protected}(source\_uri) \COMMENT{Init}
		\State $D \gets$ \algoname{Where}($D$, sex == `M' AND age $\in [30, 39]$) \COMMENT{Transform}\label{line:where}
		\State $D \gets$ \algoname{select}(salary)	\COMMENT{Transform}\label{line:select}
		\State $\mathbf{x} \gets$ \algoname{T-Vectorize}($D$)		\COMMENT{Transform}\label{line:vec}
		\State $\mathbf{P} \gets$ \ahpOp($\mathbf{x}$, $\epsilon/2$)	\COMMENT{\hl{Partition} Select}\label{line:ahp}
		\State $\bar{\mathbf{x}} \gets$ \algoname{\reduceOp}($\mathbf{x}$, $\mathbf{P}$)				\COMMENT{Transform}\label{line:reduce}
		\State $\mathbf{M} \gets$ \algoname{Identity}($|\bar{\mathbf{x}}|$) 		\COMMENT{Query Select}\label{line:identity}
		\State $\mathbf{y} \gets$ \algoname{VecLaplace}($\bar{\mathbf{x}}$, $\mathbf{M}$, $\epsilon/2$) 	\COMMENT{Query}\label{line:laplace}
		\State $\hat{\mathbf{x}}$ $\gets$ \algoname{NNLS}($\mathbf{P}$, $\mathbf{y}$)			\COMMENT{Inference}\label{line:inference}
		\State $\mathbf{W_{pre}} \gets$ \algoname{Prefix}($|\mathbf{x}|$)	 \COMMENT{Query Select}\label{line:prefix}
		\State
		\Return $\mathbf{W_{pre}} \cdot \hat{\mathbf{x}}$	\COMMENT{Output}\label{line:output}
	\end{algorithmic}
\end{algorithm}

%
%
First, we describe an example \sys plan and use it to introduce the different operator classes.  \cref{alg:cdf} shows the pseudocode for a plan authored in \sys, which takes as input a table $D$ with schema [Age, Gender, Salary] and returns the differentially private estimate of the empirical cumulative distribution function (CDF) of the Salary attribute, for males in their 30's.  The plan is fairly sophisticated and works in multiple steps.  First the plan uses \emph{transformation} operators on the input table $D$ to filter out records that do not correspond to males in their 30's (\cref{line:where}), selecting only the salary attribute (\cref{line:select}). Then it uses another transformation operator to construct a vector of counts $\mathbf{x}$ that contains one entry for each value of salary. $x_i$ represents the number of rows in the input (in this case males in their 30's) with salary equal to $i$.


Before adding noise to this histogram, the plan uses a \emph{\hl{partition} selection} operator, \ahpOp (\cref{line:ahp}). \hl{Operators in this class choose a partition of the data vector which is later used in a transformation}. \ahpOp uses the sensitive data to identify a \hl{partition $\mathbf{P}$ of the counts in $\mathbf{x}$ such that counts within a partition group are close}. Since \ahpOp uses the input data, it expends part of the privacy budget (in this case $\epsilon/2$). \ahpOp is a key subroutine in AHP \cite{zhangtowards}, which was shown to have state-of-the-art performance for histogram estimation \cite{Hay16Principled}.

Next the plan uses \reduceOp (Line~\ref{line:reduce}), another transformation operator on  $\mathbf{x}$, to apply the partition $\mathbf{P}$ computed by \ahpOp. This results in a new reduced vector $\bar{\mathbf{x}}$ that contains one entry for each partition group in $\mathbf{P}$ and the entry is computed by adding up counts within each group.

The plan now specifies a set of measurement queries $\mathbf{M}$ on $\bar{\mathbf{x}}$ using the \algoname{Identity} \emph{query selection} operator (\cref{line:identity}). The identity matrix corresponds to querying all the entries in $\bar{\mathbf{x}}$ (since $\mathbf{M} \bar{\mathbf{x}} = \bar{\mathbf{x}}$). Query selection operators do not answer any query, but rather specify which queries should be estimated.  (This is analogous to how \hl{partition} selection operators only select a partition but do not apply it.)  Next, \algoname{Vector Laplace} returns differentially private answers to all the queries in $\mathbf{M}$.  It does so by automatically calculating the sensitivity of the vectorized queries -- which depends on all upstream data transformations -- and then using the standard Laplace mechanism (\cref{line:laplace}) to add noise.  This operator consumes the remainder of the privacy budget (again $\epsilon/2$).


So far the plan has computed an estimated histogram of partition group counts $\mathbf{y}$, while our goal is to return the empirical CDF on the original salary domain. Hence, the plan uses the noisy counts on the reduced domain $\mathbf{y}$ to infer non-negative counts in the original vector space of $\mathbf{x}$ by invoking an \emph{inference} operator \algoname{NNLS} (short for non-negative least squares) (\cref{line:inference}). \algoname{NNLS}$(\mathbf{P}, \mathbf{y})$ finds a solution, $\hat{\mathbf{x}}$, to the problem $\mathbf{P}\hat{\mathbf{x}} = \mathbf{y}$, such that all  entries of  $\hat{\mathbf{x}}$ are non-negative. Lastly, the plan constructs the set of queries, $\mathbf{W_{pre}}$, needed to compute the empirical CDF (a lower triangular $n \times n$ matrix representing prefix sums) by calling the query selection operator \algoname{Prefix}($n$) (\cref{line:prefix}), and returns the output  (\cref{line:output}).

\subsection{\sys design principles}  \label{sec:benefits}

The design of \sys is guided by the following principles.  With each principle, we include references to future sections of the paper where the consequent benefits are demonstrated.



\begin{description} \itemsep 1ex

\item[Expressiveness] \sys is designed to be expressive, meaning that a wide variety of state-of-the-art algorithms can be written succinctly as \sys plans.  To ensure expressiveness, we carefully designed a foundational set of operator classes that cover features commonly used by leading differentially private algorithms. We demonstrate the expressiveness of our operators by showing in \cref{sec:plans} that the algorithms from the recent DPBench benchmark \cite{Hay16Principled} can be readily re-implemented in \sys.

\item[Privacy ``for free'']
\sys is designed so that any plan written in \sys automatically satisfies differential privacy. The formal statement of this privacy property is in \cref{sec:sub:privacy}.  This means that plan authors are not burdened with writing privacy proofs for each algorithm they write.
Furthermore, when invoking privacy-critical operators that take noisy measurements of the data, the magnitude of the noise is automatically calibrated.  As described in \cref{sec:framework}, this requires tracking all data transformations and measurements and using this information to handle each new measurement request.

\item[Reduced privacy verification effort]
%
Ensuring that an algorithm implementation satisfies differential privacy requires verifying that it matches the algorithm specification.  The design of \sys reduces the amount of code that must be vetted each time an algorithm is crafted.  First, since an algorithm is expressed as a plan and all plans automatically satisfy differential privacy, the code to be vetted is solely the individual operators.  Second, operators need to be vetted only once but may be reused across multiple algorithms. Finally, it is not necessary to vet every operator, but only the privacy-critical ones (as discussed in \cref{sec:framework}, \sys mandates a clear distinction between privacy-critical and non-private operators).  This means that verifying the privacy of an algorithm requires checking fewer lines of code.  In \cref{sec:plans}, we compare the verification effort to vet the DPBench codebase\footnote{Available at: \url{https://github.com/dpcomp-org/dpcomp_core}} against the effort required to vet these algorithms when expressed as plans in \sys.



\item[Transparency]
In \sys, all algorithms are expressed in the same form: each is a plan, consisting a sequence of operators where each operator is selected from a class of operators based on common functionality.  This facilitates algorithm comparison and makes differences between algorithms more apparent.
In \cref{sec:plans}, we summarize the plan signatures of a number of state-of-the-art algorithms (pictured in \cref{fig:plan_signatures}).  These plan signatures reveal similarities and common idioms in existing algorithms.  These are difficult to discover from the research literature or through code inspection.

\item[Efficiency and Scalability] Many \sys plans compute on data vectors formed from projections of an input table.  The current implementation of \sys relies on storing these vectors in memory on a single machine.  Even under this restriction, it is challenging to get all \sys operators to run efficiently.  Our specialized matrix representation techniques, presented in \cref{sec:implementation}, allow many of the key operators to scale to large data vectors without imposing undue restrictions on plan authors.

\end{description}

We believe that \sys, by supporting the design principles described above, provides an improved platform for designing and deploying differentially private algorithms.



\section{Preliminaries} \label{sec:background}



The input to \sys is a database instance of a single-relation schema
$T(A_1, A_2, \ldots, A_{\ell})$.
Each attribute $A_i$ is assumed to be discrete (or suitably discretized).  A {\em condition formula}, $\phi$, is a Boolean condition that can be evaluated on any tuple of $T$. We use $\phi(T)$ to denote the number of tuples in $T$ for which $\phi$ is true.  A number of operators in \sys answer linear queries over the table. A linear query is the linear combination of any finite set of condition counts:
\begin{definition}[Linear counting query (declarative)]  A linear query $q$ on $T$ is defined by conditions $\phi_1 \dots \phi_k$ and coefficients $c_1 \dots c_k \in \mathbb{R}$ and returns $q(T)=c_1 \phi_1(T) + \dots + c_k \phi_k(T)$.
\end{definition}
It is common to consider a vector representation of the database, denoted $\mathbf{x} = [x_1 \dots x_n]$, where $x_i$ is equal to the number of tuples of type $i$ for each possible tuple type in the relational domain of $T$.  The size of this vector, $n$, is the product of the attribute domains.
Then it follows that any linear counting query has an equivalent representation as a vector of $n$ coefficients, and can be evaluated by taking a dot product with $\mathbf{x}$.  Abusing notation slightly, let $\phi(i) = 1$ if $\phi$ evaluates to true for the tuple type $i$ and 0 otherwise.
\begin{definition}[Linear counting query (vector)]
For a linear query $q$ defined by $\phi_1 \dots \phi_k$ and $c_1 \dots c_k$, its equivalent vector form is $\q = [q_1 \dots q_n]$ where $q_i = c_1 \phi_1(i) + \dots + c_k \phi_k(i)$. The evaluation of the linear query is $\q \cdot \mathbf{x}$, where $\mathbf{x}$ is vector representation of $T$.
\end{definition}
In the sequel, we will use vectorized representations of the data frequently. We refer to the {\em domain} as the size of $\mathbf{x}$, the vectorized table.  This vector is sometimes large and a number of methods for avoiding its materialization are discussed later.

Let $\db$ and $\db'$ denote two tables of the same schema, and let $\db \oplus \db' = (\db - \db') \cup (\db' - \db)$ denote the symmetric difference between them. We say that $\db$ and $\db'$ are neighbors if $|\db \oplus \db'| = 1$.

\begin{definition}[Differential Privacy~\cite{dwork2006calibrating}] \label{def:diffp}
A randomized algorithm $\algG$ is $\epsilon$-differentially private if for any two instances $\db$, $\db'$ such that $|\db \oplus \db'| = 1$, and any subset of outputs $S \subseteq Range(\algG)$,
\[
Pr[ \algG(\db) \in S] \leq \exp(\epsilon) \times Pr[ \algG(\db') \in S]
\]
\end{definition}

Differentially private algorithms can be composed with each other and other algorithms using composition rules, such as sequential and parallel composition~\cite{mcsherry2009pinq} and post-processing~\cite{Dwork14Algorithmic}.
Let $f$ be a function on tables that outputs real numbers. The {\em sensitivity} of the function is defined as: $max_{|\db \oplus \db'|=1} |f(\db) - f(\db')|$.



\begin{definition}[Stability] \label{def:stabiliy}
Let $g$ be a transformation function that takes a data source (table or vector) as input and returns a new data source (of the same type) as output.
For any pair of sources $S$ and $S'$ let $|S \oplus S'|$ denote the distance between sources.  If the sources are both tables, then this distance is the size of the symmetric difference; if the sources are both vectors, then this distance is the $L_1$ norm; if the sources are of mixed type, it's undefined.
Then the stability of $g$ is: $\max_{S, S': |S \oplus S'| = 1} |g(S) \oplus g(S')|$. When the stability of $g$ is at most $c$ for some constant $c$, we say that $g$ is $c$-stable.

\end{definition}

\section{Execution Framework And Privacy Enforcement} \label{sec:framework}

This section describes the execution environment and then formalizes the claim that any program executed in \sys satisfies differential privacy.
%
%

\subsection{Protected Kernel and Client Space} \label{sec:kernel}


The execution framework consists of an untrusted client space and a \pk that encloses the private data.  An \sys program, which we call a plan, runs in the unprotected client space.
When the plan needs to interact with the private data, it does so through privileged operators that can issue requests to the \pk.  Such operators may, for example, request that \pk apply a data transformation or perhaps return a noisy measurement.
The \pk
services requests from privileged operators, only executing them if their cost is within the available privacy budget.
The distinction between the client space and the \pk is a fundamental one in \sys.  It allows authors to write plans that consist of operator calls embedded in otherwise arbitrary code (which may freely include conditionals, loops, recursion, etc.).

The \pk is initialized by specifying a single protected data object---an input table $\db$---and a global privacy budget, which we denote as $\budget$.  Note that requests for data transformations may cause the \pk to derive additional data sources.  Thus, the \pk maintains a {\em data source environment}, which consists of a mapping between data source variables, which are exposed to the client, and the protected data objects, which are kept private.  In addition, the data source environment tracks the transformation lineage of each data source.  It also maintains the stability of each transformation (defined in~\cref{sec:background}).
%
Note that in describing operators (\cref{sec:operators}), we speak informally of operators having data sources as inputs and outputs rather than data source {\em variables}.  A layer of indirection is always maintained in the implementation but sometimes elided in our descriptions to simplify the presentation.

\subsection{Operator types} \label{sub:op_types}

Operators have one of three types, based on their interaction with the \pk.  The first type is a \textbf{\privOp} operator, which requests that the \pk perform some action on the private data (e.g., a transformation) but receives only an acknowledgement that the operation has been performed.  The second type is a \textbf{\privPubOp} operator, which receives information about the private data (e.g., a measurement) and thus consumes privacy budget.  The last type is a \textbf{\pubOp} operator, which does not interact with the \pk at all and can be executed entirely in client space.  An example of a public operator would operators that perform inference on the noisy measurements received from the \pk.
When describing operators in~\cref{sec:operators}, we color code them based on their type.

\subsection{Privacy Guarantee} \label{sec:sub:privacy}

\def\tran{r_k}
\def\tranvar{R_k}
\def\alltran{\mathcal{R}_k}

In this section, we state the privacy guarantee offered by \sys.  Informally, \sys  ensures that if the \pk is initialized with a source database $\db$ and a privacy budget $\budget$, then any plan (chosen by the client) will satisfy $\budget$-differential privacy with respect to $\db$.  Note that if the client exhausts the privacy budget, subsequent calls to {\privPubOp} operators will return an exception, indicating that they are not permitted.  Importantly, an exception itself does not leak sensitive information -- i.e., the decision to return an exception does not depend on the private state.

A {\em transcript} is a sequence of operator calls and their responses. Formally, let $\tran = \langle op_1, a_1, \dotsc, op_k, a_k \rangle$ denote a length $k$ sequence where $op_i$ is an operator call and $a_i$ the response.
%
We assume that the value of $op_i$ is a deterministic function of $a_1, \dots, a_{i-1}$.
%
We use $\tranvar = \tran$ to denote the event that the first $k$ operations result in transcript $\tran$.
Let $\alltran$ be the set of all possible transcripts of length $k$.  We assume that all {\privPubOp} operators output values from an arbitrary, but finite set.  Thus, the set of possible transcripts is finite.
%
%
Let $P(\tranvar = \tran \;|\; \init(\db, \budget))$ be the conditional probability of event $\tranvar = \tran$ given that the system was initialized with input $\db$ and a privacy budget of $\budget$.
\begin{theorem}[Privacy of \sys plans] \label{thm:privacy_of_plan_transcripts}
Let $\db, \db'$ be any two instances such that $|\db \oplus \db'| = 1$.  For all $k \in \mathbb{N}^+$ and $\tran \in \alltran$,
{\small\[
P(\tranvar = \tran \;|\; \init(\db, \budget)) \leq \exp(\budget) \times P(\tranvar = \tran \;|\; \init(\db', \budget)).
\]}
\label{thm:plan_privacy}
\end{theorem}
The proof of \cref{thm:plan_privacy}, which appears in the sequel, extends the proof in~\cite{Ebadi17Featherweight} to support the \partitionOp operator.

While \sys ensures differential privacy, private information could be leaked via side-channel attacks (e.g., timing attacks).  Privacy engineers who design operators are responsible for protecting against such attacks; an analysis of this issue is beyond the scope of this paper.


\subsection{Privacy Proof} \label{sec:privacy_proof}

This section presents a proof of~\cref{thm:privacy_of_plan_transcripts}.  We start by introducing some supporting concepts and notation.  (Some notation is adapted from~\cite{Ebadi17Featherweight}.)

\paragraph{Information tracked by the \pk} The \pk maintains the following state, which we denote as $\api$:
\squishlist
	\item A set of source variables $SV$.
	\item A data source environment $E$ maps each source variable $sv \in SV$ to an actual data source $S$, as in $E(sv) = S$.  (Recall that sources can be tables or vectors.)
	\item A transformation graph: the nodes are $SV$ and there is an edge from $sv$ to $sv'$ if $sv'$ was derived via transformation from $sv$.  (Note: a partition transformation introduces a special dummy data source variable whose parent is the source variable being partitioned and whose children are the variables associated with each partition.)
	\item A stability tracker $St$ maps each source variable $sv \in SV$ to a non-negative number: $St(sv)$ represents the stability (\cref{def:stabiliy}) of the transformation that derived data source $sv$ from the initial source, or 1 if $sv$ is the initial source.
	\item A budget consumption tracker $B$ that maps each source variable $sv \in SV$ to a non-negative number: $B(sv)$ represents the total budget consumption made by queries to $sv$ or \emph{to any source derived from $sv$}.
	\item A query history $\QH$ that maps each source variable to information about the state of queries asked about $sv$ or any of its descendants.  Specifically, for $sv \ in SV$, $\QH(sv)$ returns of a set of tuples $(q, s, \sigma, v)$ where the meaning of the tuple is that query $q$ was executed on data source $s$ (which is $sv$ or one of its descendants) with $\sigma$ noise, the result was $v$.  In the context of the proof a query is any \privPubOp operator.  Such an operator is assumed to satisfy $\epsilon$-differential privacy with respect to the data source on which it is applied.
	\item The global privacy budget, denoted $\budget$.
\squishend

When the \pk is initialized, as in $\init(T, \budget)$, it sets global budget to $\budget$, creates new source variable $sv_{root}$, sets $ \allowbreak E(sv_{root}) = T$, sets $St(sv_{root}) = 1$, and $B(sv_{root}) = 0$, and adds $sv_{root}$ to the transformation graph.

\paragraph{Budget Management}
When a query request is issued to the \pk, the \pk uses \cref{alg:requests} to check whether the query can be answered given the available privacy budget.

\newcommand{\peq}{\mathrel{+}=}
\begin{algorithm}
\caption{An algorithm for budget requests}\label{alg:requests}
\begin{algorithmic}[1]
\Procedure{Request}{$sv$, $\sigma$}
\If {$sv$ is the root}
\State If $B(sv) + \sigma > \budget$, return \textsc{False}.  Otherwise $B(sv) \peq \sigma$ and return \textsc{True}.
\ElsIf {$sv$ is a partition variable}
\State Let $sv_{child}$ be the child from which the request came..
\State Let $r = \max \set{ B(sv_{child}) + \sigma - B(sv), 0}$
\State Let $ans =$ \textsc{Request}(parent($sv$), $r$).
\State If $ans = $ \textsc{False}, return \textsc{False}.  Otherwise, $B(sv) \peq r$ and return \textsc{True}.
\Else
\State $ans =$ \textsc{Request}($parent(sv)$, $s \cdot \sigma$)
\Comment{$s$ is stability factor of $sv$ wrt its parent}
\State if $ans = $ \textsc{False}, return \textsc{False}.
\State $B(sv) \peq \sigma$.  Return \textsc{True}.
\EndIf
\EndProcedure
\end{algorithmic}
\end{algorithm}

\paragraph{Configurations}
A configuration, denoted $\C = \braket{\client, \api}$, captures the state of the client, denoted $\client$, and the state of the \pk, denoted $\api$.  The client state can be arbitrary, but state updates are assumed to be deterministic.

We can define the similarity of two configurations $\C$ and $\C'$ as follows.  (Notation: we use $X'$ to refer to component $X$ of configuration $\C'$.) We say that $\C \sim \C'$ iff $\client = \client'$ and $\api' \sim \api'$ where $\api \sim \api'$ iff $SV = SV'$ and the transformation graphs are identical and for each $sv \in SV$ the following conditions hold:
\squishlist
	\item $St(sv) = St'(sv)$, $B(sv) = B'(sv)$, $\QH(sv) = \QH'(sv)$, and $\budget = \budget'$.
	\item $|E(sv) \oplus E'(sv)| \leq St(sv) = St'(sv)$ where $|x \oplus y|$ is measured as symmetric difference when the sources $x$ and $y$ are tables and $L_1$ distance for vectors; see \cref{def:stabiliy}.)
\squishend

We introduce a lemma that bounds the difference probability between query answers.  Let $P(q(E(s), \sigma) = v)$ denote the probability that query operator $q$ when applied to data source $E(s)$ with noise $\sigma$ returns answer $v$.

\begin{lemma} \label{lemma:tree}
Let $\C \sim \C'$.  For any $sv \in SV$ with non-empty $\mathcal{Q}(sv)$, the following holds:
{ \footnotesize
\begin{align} \label{eqn:invar}
&\prod_{(q, s, \sigma, v) \in \QH(sv)}
P( q(E(s), \sigma) = v)
\\
&\leq
\exp( B(sv) \times | E(sv) \oplus E'(sv)|) \times
\prod_{(q, s, \sigma, v) \in \QH'(sv)}
P( q(E'(s), \sigma) = v) \nonumber
\end{align}
}
\end{lemma}

\begin{proof}
Proof by induction on a reverse topological order of the transformation graph.

\emph{Base case}: Consider a single $sv$ at the end of the topological order (therefore it has no children).  If $\QH(sv)$ is empty, it holds trivially.  Assume non-empty.  Consider any $(q, s, \sigma, v) \in \QH(sv)$.  Since $sv$ has no children, then $s = sv$.  Furthermore, because the only budget requests that apply to $sv$ are from direct queries, we have (according to \cref{alg:requests}),
$ B(sv) = \sum_{(q, s, \sigma, v) \in \QH(sv)} \sigma$.
Since we assume that any query operator satisfies $\epsilon$-differential privacy with respect to its source input, we have $P(q(E(s), \sigma) = v) \leq P(q(E'(s), \sigma) = v) \times \exp( \sigma \times | E(s) \oplus E'(s)|)$.  Substituting $sv$ for $s$ and taking the product over all terms in $\QH(sv)$, we get \cref{eqn:invar}.

\emph{Inductive case}: Assume \cref{eqn:invar} holds for all nodes later in the topological order.  Therefore it holds for any child $c$ of $sv$.  We can combine the inequalities for each child into the following inequality over all children,
{
\footnotesize
\begin{align*} 
&\prod_{c \in \text{children}(sv)}
\prod_{(q, s, \sigma, v) \in \mathcal{Q}(c)}
P( q(E(s), \sigma) = v) \\
&\leq
\prod_{c \in \text{children}(sv)}
\exp( B(c) \times | E(c) \oplus E'(c)|) \times
\prod_{(q, s, \sigma, v) \in \mathcal{Q}(c)}
P( q(E'(s), \sigma) = v) \\
&=\exp\left( \sum_{c \in \text{children}(sv)} B(c) \times | E(c) \oplus E'(c)| \right)\\
 & \qquad \times
\prod_{c \in \text{children}(sv)}
\prod_{(q, s, \sigma, v) \in \mathcal{Q}(c)}
P( q(E'(s), \sigma) = v)
\end{align*}
}
There are two cases, depending what type of table variable $sv$ is.

First, consider the case when $sv$ is {\em not} a special partition variable.
We know by transformation stability that $| E(c) \oplus E'(c)| \leq s \times | E(sv) \oplus E'(sv)|$ where $s$ is the stability factor for the transformation.
In addition, $\sum_{c} B(c) \times s \leq B(sv)$ because, according to \cref{alg:requests}, every time a request of $\sigma$ is made to child $c$, a request of $s \times \sigma$ is made to $sv$.
Therefore,
{
\footnotesize
\begin{align*}
\sum_{c \in \text{children}(sv)} B(c) \times | E(c) \oplus E'(c)|
&\leq  \sum_{c \in \text{children}(sv)} B(c) \times s \times | E(sv) \oplus E'(sv)|  \\
&\leq B(sv) \times | E(sv) \oplus E'(sv)|
\end{align*}
}
Furthermore, observe that each term in $(q, s, \sigma, v) \in \QH(c)$ also appears in $\QH(sv)$.  In addition, $\QH(sv)$ includes any queries on $sv$ directly (and we know from an argument similar to the base case that \cref{eqn:invar} holds for these queries).
Therefore \cref{eqn:invar} holds on $sv$.

Now, consider the case where $sv$ is the special partition variable.  Let $m = \max_{c} B(c)$.  We get the following
\begin{eqnarray*}
\lefteqn{\sum_{c \in \text{children}(sv)} B(c) \times | E(c) \oplus E'(c)|
\leq  \sum_{c \in \text{children}(sv)} m \times | E(c) \oplus E'(c)|}  \\
&&= m \times \sum_{c \in \text{children}(sv)} | E(c) \oplus E'(c)|
= m \times | E(sv) \oplus E'(sv)| \\
&&= B(sv) \times | E(sv) \oplus E'(sv)|
\end{eqnarray*}
The second to last line follows from the fact that $sv$ is partition transformation.  The last line follows from how $B(sv)$ is updated according \cref{alg:requests}.
\end{proof}

\paragraph{Main Proof}
We use $\Czero(T, \budget, P_0)$ to denote the initial configuration in which the \pk has been initialized with $\init(T, \budget)$ and the client state is initialized to $P_0$.
We use the notation $\Czero(T, \budget, P_0) \stackrel{t}{\Rightarrow}_p \C$ to mean that starting in $\Czero$ after $t$ operations, the probability of being in configuration $\C$ is $p$.

\begin{theorem} \label{thm:detailed_privacy}
If $T \sim_1 T'$ and $\Czero(T, \budget, P_0) \stackrel{t}{\Rightarrow}_p \C$ such that $B(sv_{root}) = \epsilon$ in $\C$, then $\epsilon \leq \budget$ and there exists $\C'$ such that $\Czero(T', \budget, P_0)  \stackrel{t}{\Rightarrow}_q \C'$ where $\C \sim \C'$ and $p \leq q \cdot \exp(\epsilon)$.
\end{theorem}

\cref{thm:privacy_of_plan_transcripts} follows as a corollary from \cref{thm:detailed_privacy}.

\begin{proof}
Proof by induction on $t$.

\emph{Base case}: $t=0$.  This implies that $p = q = 1$, $\epsilon = 0$, and $\C = \Czero(T, \budget, P_0)$ and $\C' = \Czero(T', \budget, P_0)$.  It follows that $\C \sim \C'$ because we are given that $T \sim_1 T'$ and the rest of the claim follows.

\emph{Inductive case}: Assume the claim holds for $t$, we will show it holds for $t + 1$.
Let $\C_1$ be any configuration such that $\Czero(T, \budget, P_0) \stackrel{t}{\Rightarrow}_{p_1} \C_1$ where in $\C_1$, we have $B(sv_{root}) = \epsilon_1$.

The inductive hypothesis tells us that $\epsilon_1 \leq \budget$ and that there exists a $\C'_1$ such that
$\Czero(T', \budget, P_0) \stackrel{t}{\Rightarrow}_{q_1} \C'_1$ and $\C_1 \sim \C'_1$ and $p_1 \leq q_1 \times \exp(\epsilon_1)$.

Because $\C_1 \sim \C'_1$, it follows that the client is in the same state and so the next operation request from the client will be the same in $\C_1$ and $\C'_1$.  The proof requires a case analysis based on the nature of the operator.  We omit analysis of transformation operators or operators that are purely on the client side as those cases are straightforward: essentially we must show that the appropriate bookkeeping is performed by the \pk.  We focus on the case where the operator is a query operator.

For a query operator, there are two cases: (a) running out of budget, and (b) executing a query.  For the first case, by the inductive hypothesis $\C_1 \sim \C'_1$ and therefore if executing \cref{alg:requests} yields False on the \pk state in $\C_1$, it will also do so on the \pk state in $\C'_1$.  For the second case, suppose query $q$ is executed on source $sv$ with noise $\sigma$ and answer $v$ is obtained. The \pk adds the correpsonding entry to the query history $\QH$.  Let $\C$ denote the resulting state.  Let $\C'$ correspond to extending $\C'_1$ in a similar way.  Thus $\C \sim \C'$.

It remains to show two things.  First, letting $B(sv_{root}) = \epsilon$, we must show that $\epsilon \leq \budget$.  This follows from \cref{alg:requests} which does not permit $B(sv_{root})$ to exceed $\budget$.
Second, we must bound the probabilities.  Suppose that the probability of this query answer in $\C$ is $p_2$ and the probability of this answer on $\C'$ is $q_2$.  It remains to show that $p_1 \cdot p_2 \leq \exp(\epsilon) \cdot q_1 \cdot q_2$.  For this we rely on \cref{lemma:tree} applied to $sv_{root}$ with the observations that the product of probabilities bounded in \cref{lemma:tree} corresponds to the probabilities in $p_1 \cdot p_2$ that do not trivially equal 1 and that $| E(sv_{root}) \oplus E'(sv_{root})| = 1$.
%
%
%
%
%
%
%
%
%
%
%
%
%
\end{proof}

\section{Operators and Operator Classes} \label{sec:operators}

We now describe in detail the operators and operator classes in \sys.  A full list of operators is shown in \cref{fig:operator_index} where they are arranged into classes and color-coded by type ({\privOp}, {\privPubOp}, or {\pubOp}).  Along with descriptions of the operator classes we explain their role in plans and prove supporting properties.






\begin{figure}[t!]
	\centering
		\includegraphics[scale=.9]{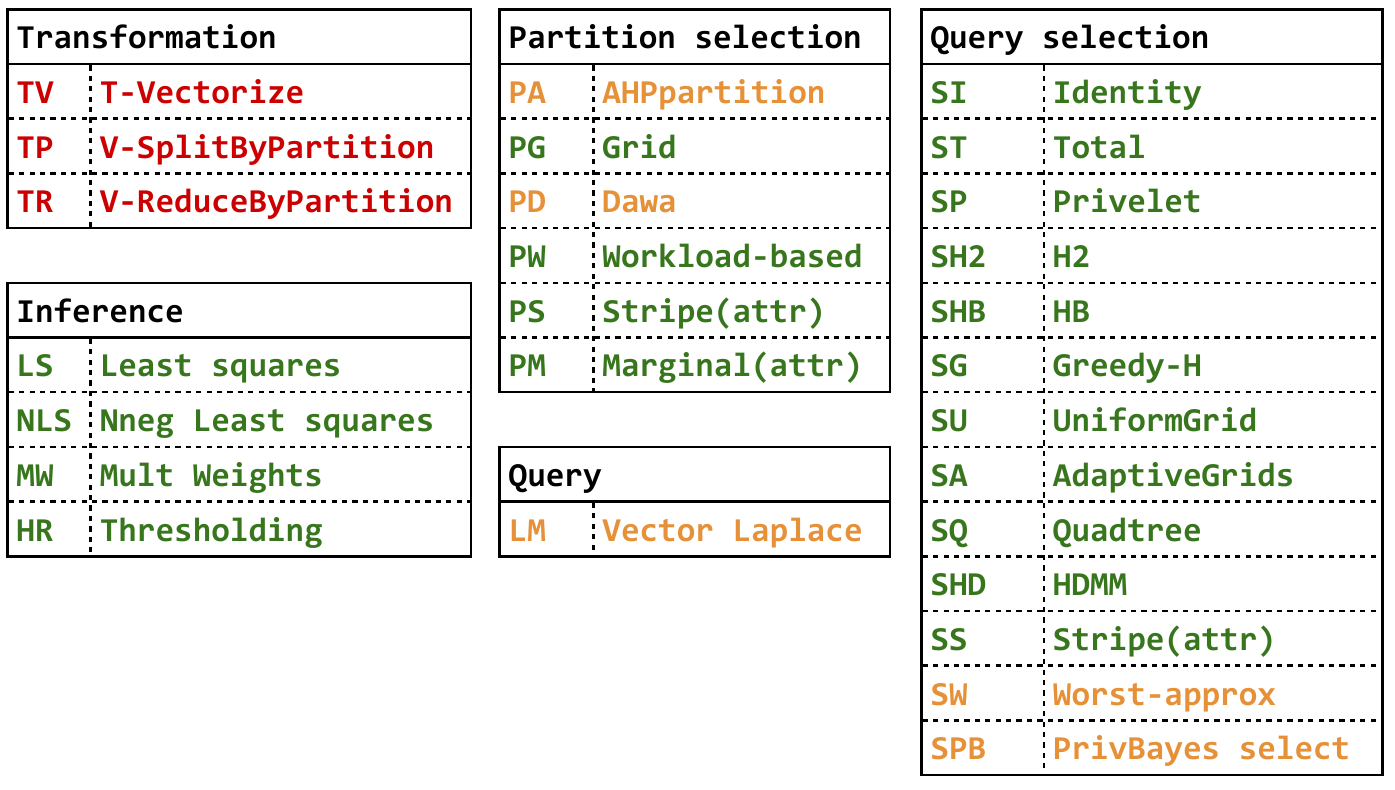}
	\caption{\label{fig:operator_index} The operators currently implemented in \sys. {\privOp} operators are {\color{BrickRed} red}, {\privPubOp} operators
	are {\color{orange} orange}, and {\pubOp} operators are {\color{OliveGreen} green}.}
\end{figure}

\subsection{Transformation Operators}


Transformation operators take as input a data source variable (either a table or a vector) and output a transformed data source (again, either a table or vector). Transformation operators modify the data held in the kernel, without returning answers. So while they do not expend privacy budget, they can affect the privacy analysis through their \emph{stability} (\cref{sec:background}).
\noindent Every transformation in \sys has a well-established stability.

\subsubsection*{Table Transformations}
\sys supports table transformations \op{Select}, \op{Where}, \tablePartitionOp{}, and \op{GroupBy}, with stabilities of 1, 1, 1, and 2 respectively. The definitions of the operators are nearly identical to those described in PINQ \cite{mcsherry2009pinq} and are not repeated here. As \sys currently handles  programs that use linear queries on single tables, the \op{Join} operator is not yet supported.

\subsubsection*{Vectorization}
All of the plans in \sys start with table transformations and typically transform the resulting table into a vector using \op{T-Vectorize} (and all later operations happen on vectors). The \op{T-Vectorize} operator is a transformation operator that takes as input a table $T$ and outputs a vector $\mathbf{x}$ that has as many cells as the number of elements in the table's domain (recall the discussion of domain \cref{sec:background}). Each cell in $\mathbf{x}$ represents the number of records in the table that correspond to the domain element encoded by the cell. \op{T-Vectorize} is a 1-stable transformation.


The vectorize operation can significantly impact the performance of the code, especially in high-dimensional cases, as we represent one cell per element in the domain. For this reason we allow  table transformations to reduce the domain size before running \op{T-Vectorize}. One of the primary reasons for working with the vector representation is to allow for inference operators downstream. Once in vector form, data can be further transformed as described next.

\subsubsection*{Vector Transformations}
\sys supports transformations on vector data sources.  Each vector transformation takes as input a vector $\mathbf{x}$ and a matrix $\mathbf{M}$ and produces a vector $\mathbf{x' = Mx}$.  The linearity of vector transformations is an important feature that is leveraged by downstream inference operators.  The stability of vector transformations is equal to the largest $L_1$ column norm of  $\mathbf{M}$.

The \reduceOp operator is a 1-stable vector transformation operator that reduces the dimensionality of the data vector $\mathbf{x}$ by eliminating cells from $\mathbf{x}$ or grouping together cells in $\mathbf{x}$. Such transformations are useful to (a) filter out parts of the domain that are uninteresting for the analyst,  (b) reduce the size of the $\mathbf{x}$ vector so that algorithm performance can be improved, and (c) reduce the number of cells in $\mathbf{x}$ so that the amount of noise added by measurement operators is reduced.
\eat{The input to this operator is a {\em reduce mapping} that maps each cell in $\mathbf{x}$ to at most one value in the reduced domain. \am{filtering implies this is no longer a function}.
	There are two ways to think about this.
	Reduce merges groups of cells of $\mathbf{x}$ by summing their counts, and these groups are defined by the reduce mapping. Alternately, a}

\reduceOp \hl{takes as input a partition defining a grouping of the cells in the $\mathbf{x}$.  It can be carried out by representing the partition as a $(\psize \times n)$ matrix $\P$ where $n$ is the number of cells in $\mathbf{x}$, $\psize$ is the number of groups in the partition, and $P_{ij} = 1$ if cell $j$ in $\mathbf{x}$ is mapped to group $i$, and 0 otherwise.}

The \partitionOp operator is the vector analogue of the tabular \tablePartitionOp operator. \hl{It  takes as input a partition and splits the data vector $\mathbf{x}$ into $k$ vectors, $\mathbf{x}^{(1)}, \dots, \mathbf{x}^{(k)}$, each representing a disjoint subset of the original domain.}  This operator allows us to create different subplans for disjoint parts of the domain.  This is a 1-stable vector transform.  (Note:  \partitionOp  can be expressed as $k$ linear transforms with matrices that select the appropriate elements of the domain for each partition.)

\subsection{Query Operators}
Query operators are responsible for computing noisy answers to queries on a  data source.  Since answers are returned, query operators necessarily expend privacy budget.  Query operators take a data source variable and $\epsilon$ as input.

For tables, the \op{NoisyCount} operator takes as input a table $D$ and $\epsilon$ and returns $|D| + \eta$, where $\eta$ is drawn from the Laplace distribution with scale $1/\epsilon$.
For vectors, the \op{Vector Laplace} operator takes as input a vector $\mathbf{x}$, epsilon, and a set of linear counting queries $\Q$ represented in matrix form.  Let $\Q$ be a matrix of size $(\qsize \times n)$.  \op{Vector Laplace} returns $\Q \mathbf{x} + \frac{\sigma(\Q)}{\epsilon}\vect{b}$ where $\vect{b}$ is a vector of $\qsize$ independently drawn Laplace random variables with scale $1$ and $\sigma(\Q)$ is the maximum $L_1$ norm of the columns of $\Q$.

For both query operators, it is easy to show they satisfy  $\epsilon$-differential privacy with respect to their data source input~\cite{mcsherry2009pinq,li2015matrix}.  Note, however, in the case  the source is derived from other data sources through transformation operators, the total privacy loss could be higher.  The cumulative privacy loss depends on the stability of the transformations and is tracked by the \pk.

\subsection{Query Selection Operators}


Since each query operation consumes privacy budget, the plan author must be judicious about what queries are being asked.  Recent privacy work has shown that if the plan author's goal is to answer a \emph{workload} of queries, simply asking these queries directly can lead to sub-optimal accuracy (e.g., when workload queries ask about overlapping regions of the domain).  Instead, higher accuracy can be achieved by designing a {\em query strategy}, a collection of queries whose answers can be used to reconstruct answers to the workload.  This approach was formalized by the matrix mechanism \cite{li2015matrix}, and has been a key idea in many algorithms~\cite{chaopvldb12,Li14Data-,Qardaji13Understanding,hay2010boosting,xiao2010differential,Cormode11Differentially,mckenna2018optimizing}.  Among these the recent HDMM algorithm is notable because it uses an optimization-based approach to find the query strategy that most-effectively answers the workload.  HDMM effectively scales to multi-dimensional domains, and offers state-of-the-art utility on many workloads \cite{mckenna2018optimizing}. 


A query selection operator is distinguished by its output type: a set of linear counting queries $\Q$ represented in matrix form (i.e., the matrix input to the \op{Vector Laplace} operator described above). 
As~\cref{fig:operator_index} indicates, \sys supports a large number of query selection operators, most of which are extracted from algorithms proposed in the literature. While these operators agree in terms of their output, they vary in terms of their input: some employ fixed strategies that depend only on the size of $\mathbf{x}$ (e.g., \op{Identity} and \op{Prefix} in \cref{alg:cdf}), some adapt to the workload (e.g., \op{Greedy-H}), some depend on prior measurements (e.g., \op{AdaptiveGrids}), etc.

Most query selection operators only rely on non-private information (domain size, workload) and therefore are of \pubOp type.  But there are a few that consult the private data, and thus have the \privPubOp type.  For example, \op{Worst-approx} is an operator that picks the query from a workload that is the worst approximated by a current estimate of the data. Such an operator is used by iterative algorithms like MWEM~\cite{hardt2012a-simple}.  Another is \op{PrivBayes select}, an operator that privately constructs a Bayes net over the attributes of the data source, and then returns a matrix corresponding to the
sufficient statistics for fitting the parameters of the Bayes net. This was used as a subroutine in PrivBayes~\cite{Zhang2014}.


\subsection{\hl{Partition Selection Operators}}

Partition selection operators compute a matrix $\P$ which can serve as the input to the \reduceOp and \partitionOp operators described earlier. Of course the matrix $\P$ must be appropriately structured to be a valid partition of $\x$.

This is an important operator class since much of recent innovation into state-of-the-art algorithms for answering histograms and range queries has used partitions to either  reduce the domain size of the data vector by grouping together cells with similar counts, or  split the data vector into smaller vectors and leverage the parallel composition of differential privacy to process each subset of the domain independently.
%
\sys includes partition selection operators \ahpOp and \op{Dawa} which are subroutines from the AHP~\cite{zhangtowards} and DAWA~\cite{Li14Data-} algorithms, respectively.  Both of these operators are data adaptive, and hence are \privPubOp. 
We also introduce new partition selection operators, \op{Workload-based} and \op{Stripe}, described in \cref{sec:sub:wkld,sec:census_plans} respectively.

\subsection{Inference Operators}\label{sec:inference}
An inference operator derives new estimates to queries based on the history of transformations and query answers. Inference operators never use the input data directly and hence are \pubOp. Plans typically terminate with a call to an inference operator to estimate a final set of query answers reflecting all available information computed during execution of the plan.  Some plans may also perform inference as the plan executes.

Ideally, an inference method should: (i) properly account for measurements with unequal noise; (ii) support inference over incomplete measurements (in which derived answers are not completely determined by available measurements); (iii) should incorporate all available information (including a prior or constraint on the input dataset); and lastly, (iv) inference should efficiently scale to large domains.  Many versions of inference have been considered in the literature \cite{hay2010boosting,Li:2010Optimizing-Linear,Lee15Maximum,hardt2012a-simple,Acs2012compression,zhangtowards,Qardaji13Understanding,proserpio2012workflow,Williams2010} but none meet all of the objectives above.   \sys currently supports multiple inference methods, in part to support algorithms from past work and in part to offer necessary tradeoffs among the properties above.

All the inference operators supported in \sys take as input a set of  queries, represented as a matrix $\Q$, and noisy answers to these queries, denoted $\y$.  The output of inference is a data vector $\xhat$ that best fits the noisy answers---i.e., an $\xhat$ such that $\Q \xhat\approx \y$.
The estimated $\xhat$ can then be used to derive an estimate of any linear query $\q$ by computing $\q\cdot\xhat$.
%
The inference operator may optionally take as input a set of weights, one per query (row) in $\Q$ to account for queries with different noise scales.

\sys supports two variants of least squares inference,
the most widely used form of inference in the current literature~\cite{hay2010boosting,Li:2010Optimizing-Linear,Qardaji13Understanding}.  \sys extends these methods and formulates them as general operators, allowing us to replicate past algorithms, and consider new forms of inference that support constraints.
The first variant solves a classical least squares problem:
%
\begin{definition}[Ordinary least squares (LS)]
	\begin{align} \label{ls}
	\xhat = \argmin_{x \in \mathbb{R}^n} \norm{\Q \x - \y}_2
	\end{align}
\end{definition}
%
\noindent Our second variant imposes a non-negativity constraint on $\xhat$:
%
\begin{definition}[Non-negative least squares (NNLS)]
	Given scaled query matrix $\Q$ and answer vector $\y$, the non-negative least squares estimate of $\x$ is:
	\begin{align} \label{nnls}
	\xhat &= \argmin_{x \succeq 0} \norm{\Q \x - \y}_2
	\end{align}
\end{definition}

These inference methods can also support some forms of prior information, particularly if it can be represented as a linear query. For example, if the total number of records in the input table is publicly known, or other special queries have publicly available answers, they can be added as ``noisy'' answers with negligible noise scale and they will naturally incorporated into the inference process and the derivation of new query estimates.

We also support an inference method based on a multiplicative weights update rule, which is used in the MWEM \cite{hardt2012a-simple} algorithm. This inference algorithm is closely related to the principle of maximum entropy, and is especially effective when one has measured an incomplete set of queries.

\subsubsection*{Defining inference under vector transformations}
Recall that in the discussion above we describe inference as operating on a single vector $\x$ with a corresponding query matrix $\Q$.  However, plans can include an arbitrary combination of vector transformations, followed by query operators, resulting in a collection of query answers defined over various vector representations of the data. \sys handles this by taking advantage of the structure of vector transformations and query operators, both of which perform linear transformations, therefore making it possible to map measured queries back on to the original domain (i.e., a vector produced by the \op{Vectorize} operation) and perform inference there.
This allows for the most complete form of inference but other alternatives are conceivable, for example by performing inference locally on transformed vectors and combining inferred queries.  This might have efficiency advantages, but would likely sacrifice accuracy, and is left for future investigation.

\subsubsection*{Inference: impact on accuracy}
Because inference is an operator in \sys\, algorithm authors are encouraged to use inference consistently, using all available measurements, even if they are measured in different parts of a plan.  In contrast, some existing algorithms use inference in an ad-hoc manner, performing inference on one set of measurements separately from another set of measurements.  As we show below, for unbiased plans, this is always sub-optimal and \sys\ helps to relieves the algorithm designer of the complexity of integrating measured information properly.  The following theorem follows the intuition that any unbiased noisy measurement provides information about the true data that can lower error, in expectation: 
\begin{theorem}
Given any (full rank) matrix $\Q$ of linear measurements and any linear query $\q$, the expected error of $\q$ is never higher if we include additional linear measurements using least squares inference.
\end{theorem}
\begin{proof}
Assuming all measurments have variance $1$, the expected error on a query $\q$ is $ \text{Error}_{\Q}(\q) = \q (\Q^T \Q)^{-1} \q^T $ \cite{Li:2010Optimizing-Linear}.  In general, the variance of the measurements depends on the privacy budget and the sensitivity of $\Q$, but they can always be scaled to have variance $1$.  If we augment $\Q$ with a new linear query $\vect{b}$, it becomes $ \Q' = \begin{bmatrix} \Q \\ \vect{b} \end{bmatrix} $.  We can write $ \Q'^T \Q' = \Q^T \Q + \vect{b}^T \vect{b} $ where $ \vect{b}^T \vect{b} $ is the outer product.  Using the Sherman-Morrison formula~\cite{sherman1950}, we see that 
$$ (\Q'^T \Q')^{-1} = (\Q^T \Q)^{-1} - \frac{1}{1 + \vect{b} (\Q^T \Q)^{-1} \vect{b}^T} (\Q^T \Q)^{-1} \vect{b}^T \vect{b} (\Q^T \Q)^{-1} $$
Since $ (\Q^T \Q)^{-1} $ is positive-definite, $ \vect{b} (\Q^T \Q)^{-1} \vect{b}^T \geq 0 $ and the fraction is just some positive constant $c$.  With some algebraic manipulation, we arrive at the following expression:
$$ \text{Error}_{\Q'}(\q) = \text{Error}_{\Q}(\q) - c \q (\Q^T \Q)^{-1} \vect{b}^T \vect{b} (\Q^T \Q)^{-1} \q^T $$
Letting $ v = \q (\Q^T \Q)^{-1} \vect{b}^T $, we can write that as $ \text{Error}_{\Q'}(\q) = \text{Error}_{\Q}(\q) - c v^2 $ since $ (\Q^T \Q)^{-1} $ is symmetric.  Clearly $ c v^2 $ is non-negative, so $\text{Error}_{\Q'}(\q) \leq \text{Error}_{\Q}(\q)$.  This completes the proof.   
\end{proof}


\section{Expressing known algorithms} \label{sec:plans}
As stated in \cref{sec:advantages}, \sys represents differentially private algorithms as plans composed over a rich library of operators, and supports not only simple linear sequences but also plans with iteration, recursion and branching. To highlight the expressiveness of \sys, we re-implemented state-of-the-art algorithms as \sys plans. Once the necessary operators are implemented, the plan definition for an existing algorithm is typically a few lines of code for combining operators and managing parameters.  We performed extensive testing to confirm that reimplementations in \sys of existing algorithms provide statistically equivalent outputs.

\begin{figure}[ht]
    \centering
        \includegraphics[scale=0.9]{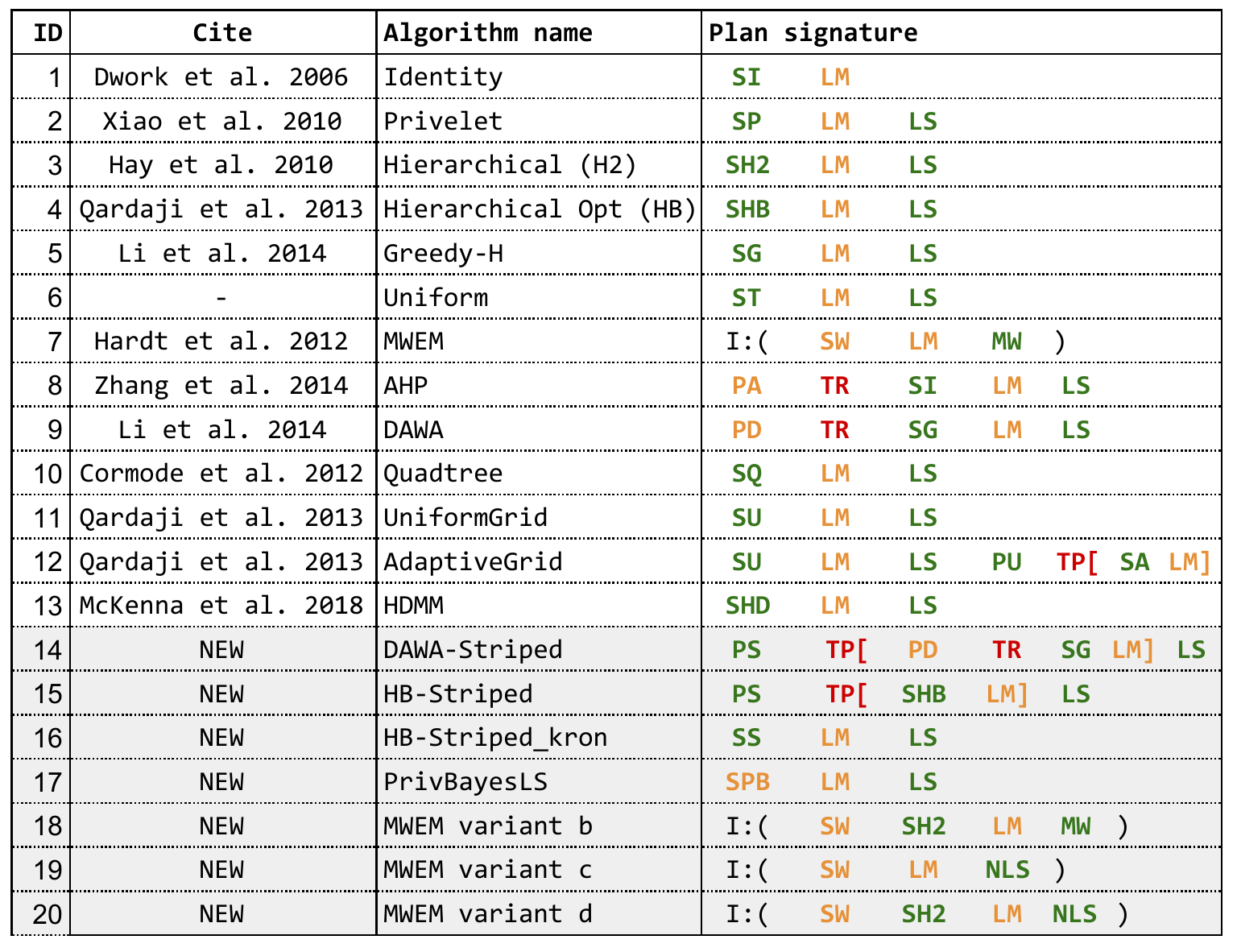}
    \caption{\label{fig:plan_signatures} The high-level signatures of plans implemented in \sys (referenced by ID).  All plans begin with a vectorize transformation, omitted for readability.  We also omit parameters of operators, including $\epsilon$ budget shares.  {\tt I}({\em subplan}) refers to iteration of a subplan and {\tt TP}[{\em subplan}] means that {\em subplan} is executed on each partition produced by \tt{TP}.}
\end{figure}

We examined 12 differentially private algorithms for answering low dimensional counting queries that were deemed competitive\footnote{This is the subset of algorithms that offered the best accuracy for at least one of the input settings of the benchmark.} in a recent benchmark study \cite{Hay16Principled}, and one new algorithm published after the benchmark study \cite{mckenna2018optimizing}. \planref {1-13} in \cref{fig:plan_signatures} abstract their \sys implementations as {\em plan signatures} where operators are represented using colored abbreviations.

\subsection{Re-implementing existing algorithms}
The algorithms are listed roughly in the order in which they were proposed in the literature and reflect the evolution of increasingly complex algorithmic techniques. The simplest Algorithm, \op{Identity} \cite{dwork2006calibrating}, is a natural application of the Laplace mechanism. It simply measures each component of the data vector. Algorithms 2 through 5 reflect the evolution of more sophisticated measurements selection, targeted toward specific workloads.  Many of these techniques were originally designed to support range queries (a small subclass of linear queries) over one- or two-dimensional data.  \op{Privelet}\cite{xiao2010differential} uses a Haar wavelet as its measurements, which allows for sensitivity that grows logarithmically with the domain size, yet allows accurate reconstruction of any range query.  The \op{Hierarchical} (H2) technique uses measurements that form a binary tree over the domain, achieving effects similar to the wavelet measurements.  \op{Quadtree} \cite{Cormode11Differentially} is the 2-dimensional realization of the hierarchical structures. All the algorithms above follow similar design idioms, allowing us to implement them using operator sequences of the same pattern:
{\em Query selection, Query, and Inference}.


All of the algorithms above are {\em data-independent}, with constant error rates for any input dataset. More recent algorithms are {\em data-dependent}, displaying different error rates on different inputs, often because the algorithmic techniques are adapting to features of the data to lower error.  The simplest data-dependent algorithm is \op{Uniform} which simply estimates the total number of records in the input and assumes uniformity across the data vector. This simple algorithm also follows the simple pattern.

A more complex example is the Multiplicative-Weights Exponential Mechanism (\op{MWEM}) \cite{hardt2012a-simple} which takes a workload of linear queries as input and runs several rounds of estimation, measuring one workload query in each round, and using the multiplicative update rule to revise its estimate of the data vector.  In each round, the Exponential Mechanism is used to select the workload query that is most poorly approximated using the current data vector estimate. 

\begin{algorithm}[h]
\caption{MWEM}\label{alg:MWEM}

\begin{algorithmic}[1]
\State $D \gets$  \algoname{Protected}(source\_uri)                                     \Comment{Init}
\State $\hat{x} \gets$ \algoname{T-Vectorize}($D$)                                        \Comment{Transform}
\For {$i = 1:T$}                                                                        
    \State $\mathbf{M} \gets$ \algoname{WorstApprox}($\hat{x}$, $\epsilon / 2T$)         \Comment{Query Selection}
    \State $\mathbf{y} \gets$ \algoname{VectorLaplace}($\mathbf{M}$, $\epsilon/ 2T$)     \Comment{Query}
    \State $\hat{\mathbf{x}}$ $\gets$ \algoname{MultWeights}($\mathbf{M}$, $\mathbf{y}$)\Comment{Inference}
\EndFor                                                                                    
\State                                                                                     
\Return  $\hat{\mathbf{x}}$                                                                \Comment{Output}
\end{algorithmic}
\end{algorithm}

In \cref{alg:MWEM}, we rewrite the algorithm with abstracted subroutines and find out that this algorithm can be represented as several iterations of the 
{\em Query selection, Query, and Inference} sequence. In \cref{fig:plan_signatures}, the iteration inherent to  \planref \planmwem \ (\op{MWEM}) is shown with $I:(..)$. 

Other data-dependent algorithms exploit partitioning, in which components of the data vector are merged and estimated only in their entirety, which uniformity assumption imposed within the regions.  The \op{DAWA} \cite{Li14Data-} and \op{AHP} \cite{zhangtowards} algorithms have custom partition selection methods which consume part of the privacy budget to identify approximately uniform partition blocks. The partition selection methods work by finding a grouping of the bins in a vector and are the key innovations of the algorithms. We encapsulate these subroutines as new operators in our framework (in the cases above, we added a partition selection operators \op{DawaPartition} (\hl{PD}) in \planref \plandawa \ and \op{AHPPartition} (\hl{PA}) in \planref \planahp).

\op{UniformGrid} and \op{AdaptiveGrid} \cite{qardaji2013differentially} focus on 2D data and both end up with partitioned sets of measurements forming a grid over a 2D domain.  \op{UniformGrid} imposes a static grid, while \op{AdaptiveGrid} uses an initial round of measurements to adjust the coarseness of the grid, avoiding estimation of small sparse regions.

\subsection{Re-implementation strategies}

The process of re-implementing in \sys this seemingly diverse set of algorithms consisted of breaking the algorithms down into key subroutines and translating them into operators. To summarize, the translation strategy typically falls into one of three categories.

The first translation strategy was to identify specific implementations of common differentially private operations and replace them with a single unified general-purpose operator in \sys. For instance, the Laplace mechanism (LM), which adds noise drawn from the Laplace distribution, appears in \emph{every one} of the 13 algorithms.  Noise addition can be implemented in a number of ways (e.g. calling a function in the numpy.random package, taking the difference of exponential random variables, etc.). In \sys, all these plans call the same \op{Vector Laplace} operator with a single unified sensitivity calculation.

Another less obvious example of this translation is for subroutines that infer an estimate of $\mathbf{x}$ using noisy query answers. With the exception of \algoname{Identity} and \algoname{MWEM}, each of the algorithms uses instances of least squares inference, often customized to the structure of the noisy query answers.  For instance, \algoname{Privelet} uses Haar wavelet reconstruction, hierarchical strategies like \algoname{HB}  and \algoname{DAWA} use a tree-based implementation of inference, and others like \algoname{Uniform} and \algoname{AHP} use uniform expansion. We replaced each of these custom inference methods with a single general-purpose least squares inference operator ({\color{OliveGreen}LS} operators in \cref{fig:plan_signatures}). 
It would still be possible to implement a specialized inference operator in \sys that exploited particular properties of a query set, but, given the efficient inference methods described in \cref{sec:improved_inference}, we did not find this to be beneficial.

Our second translation strategy was to identify higher-level patterns that reflect design idioms that exist across multiple algorithms.  In these cases, we replace one or more subroutines in the original code with a sequence of operators that capture this idiom.  As shown earlier, \planref 2, 3, 4, 5, 6, 10, and 11, 13 all consist of the operator sequence: Query selection, Query (LM), and Inference (LS), differing only in Query selection method.  For other algorithms, this idiom reappears as a subroutine, as in \planref \planahp \ (AHP) and \planref \plandawa \ (DAWA).

Finally, we were left with distinct subroutines of algorithms that represented key intellectual advances in the differential privacy literature.  
We encapsulate such subroutines as new operators (e.g. \hl{PA} in \planref \planahp\ (AHP) and \hl{PD} in \planref \plandawa \ (DAWA) ) in the framework. 

\subsection{Benefits}

We highlight the benefits of re-implementing known algorithms in \sys.

\paragraph{\bf \em Code reuse}
Once reformulated in \sys, nearly all algorithms use the \op{Vector Laplace} operator and least squares inference.  This means that any improvements to either of these operators will be inherited by all the plans. We show such an example in \cref{sec:improved_inference}.

\paragraph{\bf \em Reduced privacy verification effort} Code reuse also reduces the number of critical operators that must be carefully vetted.
The operators that require careful vetting are ones that consume the privacy budget, which are the \privPubOp operators in~\cref{fig:operator_index}.  These are: \op{Vector Laplace}, the \hl{partition} selection operators for both DAWA~\cite{Li14Data-} and AHP~\cite{zhangtowards}, a query selection operator used by PrivBayes~\cite{Zhang2014}, and a query selection operator used by the MWEM~\cite{hardt2012a-simple} algorithm that privately derives the worst-currently-approximated workload query. In contrast, for the DPBench code base, the entire code has to be vetted to audit the use and management of the privacy budget.
The end result is that verifying the privacy of an algorithm requires checking fewer lines of code.  For example, to verify the QuadTree algorithm in the DPBench codebase requires checking 163 lines of code. However, with \sys, this only requires vetting the 30-line \op{Vector Laplace} operator.  (Furthermore, by vetting just this one operator, we have effectively vetted 10 of the 18 algorithms in \cref{fig:plan_signatures}, since the only privacy sensitive operator these algorithms use is \op{Vector Laplace}.). When we consider all of the DPBench algorithms in \cref{fig:plan_signatures}, algorithms 1-12, verifying the DPBench implementation requires checking a total of 1837 lines of code while vetting all the privacy-critical operators in \sys requires checking 517 lines of code.

\paragraph{\bf \em Transparency}
As noted above, \sys plans make explicit the typical patterns that result in accurate differentially private algorithms.
Moreover, \sys plans help clarify the distinctive ingredients of state-of-the-art algorithms. For instance, DAWA and AHP (\planref 9 and \planref 8 respectively in \cref{fig:plan_signatures}) have the same structure but differ only in two operators: \hl{partition} selection and query selection.  

\section{Implementation: efficient matrix support} \label{sec:implementation}

Matrices and operations on matrices are central to the implementation of \sys operators but can become a performance bottleneck.  In this section we describe a set of specialized matrix representation techniques, based on the implicit definition of matrices, which allows for greater scalability as the size of the data vector grows.

We review next the types of matrix objects in \sys and then, in \cref{sec:sub:matrix_reps}, the different ways matrices can be represented including implicit matrices.  In \cref{sec:sub:matrix_methods}, we decompose the common matrix operations in \sys into a small set of primitive operations which every implicit matrix should support.  We then describe in \cref{sec:sub:matrix_construction} a general matrix type from which implicit matrices can be built, and use that matrix type, in \cref{sec:range_queries}, to implement common query selection operators in \sys.  We conclude with \cref{sec:improved_inference} describing the implementation of inference using implicit matrices.

\subsection{Matrix types and their operations}

Recall that matrices are used to represent three different objects in \sys: sets of workload queries, sets of measurement queries, and partitions of the domain.  In all cases, the matrices contain one column for each element of a corresponding data vector. In the case of both workload and measurement matrices, rows represent linear queries. A partition matrix describes a linear transformation that can be applied to a data vector or to a workload; one row describes a set of elements of the domain that will be combined after the transformation.

The key computations on each matrix type are shown in \cref{tbl:mat_operations} (in the left two columns).  Workload matrices and measurement matrices both represent sets of queries and so they share similar computations (such as query evaluation on a data vector and calculation of sensitivity) however we only do inference on measurement matrices.  Partition matrices are used to reduce and expand both workload matrices and data vectors.

In common plans, the number of rows in a workload or measurement measurement matrix can be as large or larger than $n$, the number of elements in the data vector.  Partitions have at most $n$ rows, but may still be large.  For plans operating on large data vectors, where $n$ approaches the size of memory, these matrices, in standard form, are infeasible to represent in memory and operate on. To address this, \sys provides flexible and efficient matrix capabilities that can be used for the efficient implementation of operators.

\subsection{Matrix representations: dense, sparse and implicit} \label{sec:sub:matrix_reps}
The matrix class in \sys supports matrices using a combination of the dense, sparse, and implicit matrices.  These representations differ in their space utilization, their generality, and the efficiency of the matrix operations they support.

A dense $m \times n$ matrix is the standard representation that stores $mn$ values.  Obviously, any matrix can be represented in this manner and all operations in \cref{tbl:mat_operations} are supported.  A sparse matrix stores only non-zero elements of a matrix.  Any matrix can be represented in sparse form,  but its efficiency depends on the number of nonzero entries.  Where $nnz(\A)$ denotes the number of nonzero elements in matrix $\A$, if $nnz(\A) \approx mn$ then sparse matrices do not offer any benefit, and may even be more expensive to represent than the dense representation.  However, if $nnz(\A) << mn$ there may be significant improvements to performance in using this representation.  

An implicit matrix is a virtual representation of a matrix that may not explicitly store all (or even any) of the elements of the matrix.  Because it is a virtual object, it must define appropriate methods so that computations with the implicit matrix produce correct results.  While not all matrices allow for efficient implicit representations, we have found that many of the matrices used in \sys operators have a structure that can be exploited for efficient implicit representation.  Note that implicit matrix representations are lossless: they do not approximate some dense matrix but represent it exactly.  Therefore an implicit matrix can always be materialized in sparse or dense form, although the goal is to perform computations without materialization.

As an example of an implicit matrix, recall the Prefix workload, an encoding of an empirical CDF, which was used in the example plan (\cref{alg:cdf}) of \cref{sec:advantages}:

\begin{example}[The Prefix workload: dense, sparse, and implicit] \label{example:prefix}
In dense form, the prefix workload is defined as a lower-triangular matrix containing $1$'s.  If $n=5$ we have:
$$ \mathbf{W_{pre}} =
\begin{bmatrix}
1 & 0 & 0 & 0 & 0 \\
1 & 1 & 0 & 0 & 0 \\
1 & 1 & 1 & 0 & 0 \\
1 & 1 & 1 & 1 & 0 \\
1 & 1 & 1 & 1 & 1 \\
\end{bmatrix} $$
A sparse representation of $\mathbf{W_{pre}}$ would store (a list of) only the nonzero elements of this matrix, but the space complexity of both dense and sparse representations remains $O(n^2)$.  In addition, the time complexity of computing matrix-vector products using the dense or sparse representation is $O(n^2)$.
However, the Prefix matrix can be completely specified by a single parameter, $n$, which is the only state stored for the implicit version of $\mathbf{W_{pre}}$.
Further, we can evaluate the matrix-vector product $\y = \mathbf{W_{pre}} \x$ using a simple one-pass algorithm over $\x$: first compute $y_1 = x_1$. Then, for $k = 2 \dots n$, compute $y_k = y_{k-1} + x_k $. Therefore, by representing the Prefix workload implicitly we can achieve $O(1)$ space complexity and $O(n)$ time complexity for computing matrix-vector products.
\end{example}



\begin{table}[tp]
\centering
\caption{\label{tbl:mat_operations} Types of matrix objects in \sys (workload, measurement, partition) and the key computations performed in plans, along with the primitive methods required to support each computation.}

\begin{tabular}{ll|l|l}
\multicolumn{3}{c}{\textbf{Key computations by matrix type}} & \textbf{Primitive methods} \\  \hline
\multicolumn{2}{l|}{\bf Workload matrix, W} & & \rule{0pt}{3ex} \\
 & Query evaluation & $\W \xhat$ & Matrix-vector product \\
 & $L_1$ Sensitivity & $\norm{\W}_1$ & Abs, Transpose, Matrix-vector product \\
 & $L_2$ Sensitivity & $\norm{\W}_2$ & Sqr, Transpose, Matrix-vector product \\
 & Gram Matrix & $ \W^T \W $ & Transpose, Matrix multiplication \\
 & Row indexing & $ \vect{w}_i $ & Transpose, Matrix-vector product \\
\multicolumn{2}{l|}{\bf Measurement matrix, M} & & \rule{0pt}{2ex} \\
 & Query evaluation & $\Q \x$ & Matrix-vector product \\
 & $L_1$ Sensitivity & $\norm{\Q}_1$ & Abs, Transpose, Matrix-vector product  \\
 & $L_2$ Sensitivity & $\norm{\Q}_2$ & Sqr, Transpose, Matrix-vector product \\
 & Inference (LS) & $\argmin_\x \norm{ \Q \x - \y }_2 $  & Transpose, Matrix-vector product \\
 & Inference (NNLS) & $\argmin_{\x \geq 0} \norm{ \Q \x - \y }_2 $  & Transpose, Matrix-vector product \\
 & Inference (MW) & $ \hat{\x}^{(k+1)} \propto \hat{\x}^{(k)} \odot \exp{(\vect{g}/N)}$  & Transpose, Matrix-vector product \\
 & &  $\vect{g} = 0.5 \Q^T (\Q \hat{\x}^{(k)} - \y)$  & \\
 \multicolumn{2}{l|}{\bf Partition matrix, P} & & \rule{0pt}{2ex} \\
 & Reduce workload & $\W' = \W \P^+ $  & Transpose, Matrix multiplication \\
 & Reduce data vector & $ \x' = \P \x $ & Matrix-vector product\\
 & Expand workload & $\W = \W' \P $ & Matrix multiplication \\
 & Expand data vector & $\x = \P^+ \x' $ & Transpose, Matrix-vector product \\  \hline
\end{tabular} \vspace{1ex}
\end{table}

\subsection{Computing with implicit matrices}
\label{sec:sub:matrix_methods}

Most implicit matrices require very little internal state to be stored.  The main challenge is therefore to insure that all necessary computations involving an implicit matrix can be carried out efficiently, hopefully without falling back to materialization of the  dense form of the matrix.  Before defining additional implicit matrix constructions, we review the key computations \sys matrix objects must support.

A careful examination of the operators currently implemented in \sys resulted in the list of key computations in \cref{tbl:mat_operations}, where the left two columns describe operations on matrices that commonly occur in plans.  Importantly, these plan-level matrix computations can be decomposed into just five fundamental matrix methods, which we call {\em primitive methods}: matrix-vector product, transpose, matrix multiplication, element-wise absolute value ($abs$), and element-wise square ($sqr$).  Matrix-vector product takes as input a vector and returns a vector.  The multiplication of two implicit matrices returns a new implicit matrix, as does transpose,  which returns another implicit matrix that represents the transpose linear transformation.  Element-wise $sqr$ and $abs$ both return new implicit matrices.

Below we review the key computations on \sys matrix objects in \cref{tbl:mat_operations}, how they can be decomposed into primitive methods, and implementation considerations. Our goal will then be to construct implicit matrices that efficiently support the primitive methods.

%


\begin{description} \itemsep 1ex

\item[Query evaluation and data reduction] In \sys plans, matrix-vector multiplication is used by workload and measurement matrices for query evaluation, and by partition matrices for reduction of the data vector.  

\item[Sensitivity] The sensitivity of measurement and workload matrices can be computed using a combination of primitive methods (abs, sqr, transpose, and matrix-vector product).  For a matrix $\Q$, we compute the maximum column sum of $abs(\Q)$ (for $L_1$ sensitivity) or $sqr(\Q)$ (for $L_2$ sensitivity) which can be done by doing a transpose-matrix-vector product with $\vect{1}$ -- the vector of all $1$'s:
\begin{align*}
\norm{ \Q }_1 = max( abs(\Q)^T \vect{1} ) & & & \norm{ \Q }_2 = \sqrt{max( sqr(\Q)^T \vect{1} )}
\end{align*}


\item[Inference]
The most common form of inference in \sys plans is based on least squares.  While classical solutions to the least squares problem involve matrix decompositions and computation of the pseudo-inverse, we will show that iterative methods lead to much greater scalability in combination with our implicit matrix representations.  Least squares, non-negative least squares, and multiplicative-weights inference can all be implemented using just the matrix-vector product and transpose primitive methods. We will discuss iterative inference in \cref{sec:improved_inference}.

\item[Gram Matrix]
Some workload-adaptive mechanisms like GreedyH and HDMM require the (materialized) gram matrix of the workload.  For a workload $\W$, the gram matrix is $\W^T \W$.  This computation can be implemented in terms of transpose, matrix multiplication, and materialize.  For extremely large workloads with special structure, where $\W^T \W$ is much smaller than $\W$ (like the set of all range queries), a more efficient version can be implemented directly that avoids using the primitive methods.

\item[Partition reduction and expansion]

Partition matrices need to be able to reduce the data vector and workload.  They also have to be able to do the reverse expansion operations.  As we show in the proof of Prop. \ref{prop:wbdr}, because of the special structure of partition matrices, the pseudo-inverse of any matrix $\P$ can be computed as the product of $\P^T$ and a diagonal matrix $\mathbf{D}$.  Thus, partition matrices simply need the three primitive methods: matrix-vector product, transpose, and matrix-multiplication.

\item[Row Indexing]  The MWEM algorithm and its variants (described in \cref{sec:mwem_variants})  use the worst-approximated query selection operator, which requires row indexing, or materialization of the $i^{th}$ row of a matrix.  This can be implemented in terms of the primitive methods transpose and matrix-vector product as follows: $ \vect{w}_i = \W^T \vect{e}_i $, where $\vect{e}_i$ is the $i^{th}$ column of an identity matrix.

\item[Materialize] If a plan requires working with matrices in a manner not supported by the interface of our implicit matrices, the matrix can always be materialized, at which point standard implementations of matrix methods can be employed.  As in row indexing, materialization can be performed by a sequence of matrix vector products with the columns of an identity matrix, i.e., $ \vect{A} \vect{e}_i $ for $ i = 1, \dots n $.

\end{description}

\subsection{Generalized matrix construction}
\label{sec:sub:matrix_construction}


\sys contains a matrix class, denoted $EMatrix$, that generalizes dense, sparse, and implicit matrices, supporting flexible matrix construction using a small set of specially designed {\em core matrices} which may then be combined with combining operations (union, product, and Kronecker product).  This provides a flexible and extensible mechanism for constructing a wide range of matrices.  The following grammar describes the construction of $EMatrix$ instances:

\begin{align*}
  CoreMatrix &= Identity \mid Ones \mid Prefix \mid Suffix \mid Wavelet \\
  EMatrix &=  DenseMatrix \mid SparseMatrix \mid CoreMatrix \\
  EMatrix &= Union(EMatrix, EMatrix) \\
  EMatrix &= Product(EMatrix, EMatrix) \\
    EMatrix &= Kronecker(EMatrix, EMatrix)
\end{align*}

Unless an $EMatrix$ is defined as a single $SparseMatrix$ or $DenseMatrix$, we consider it implicit, since it is not fully materialized.

\subsubsection*{Core matrices}

The $CoreMatrix$ type forms the basic building block for $EMatrix$ and each is defined implicitly.  The following are custom core matrices we designed to support \sys operators:

\begin{itemize}
\item \textbf{Identity}: Identity is the simplest building block.  It is defined as the matrix $ \mathbf{I} $ having the property that $ \mathbf{I} \v = \v $ for all vectors $\v$.  Thus, the implementation of matrix-vector product is trivial.  Similarly, transpose, abs, and sqr are simple no-ops.
\item \textbf{Ones}: Ones is the $m \times n$ matrix of all ones.  Matrix-vector products can be efficiently computed by summing up the entries of the input vector, and constructing a $m$-length vector with that value.  The transpose is a $ n \times m $ Ones matrix, and abs and sqr are simple no-ops.  \textbf{Total} is a special case of the Ones matrix where $m=1$.
\item \textbf{Prefix and Suffix}: The description of prefix and the algorithm for efficiently computing matrix-vector products is given in Example \ref{example:prefix}.  The transpose of Prefix is Suffix, and abs and sqr are simple no-ops.
\item \textbf{Wavelet}: Wavelet is the Haar wavelet transform.  Efficient algorithms exist for evaluating matrix-vector products implicitly with the Haar wavelet \cite{xiao2010differential}.  The transpose has a similar form, but abs and sqr, if needed, must materialize the matrix.  However, for this matrix sensitivity may be computed directly, without going through abs and sqr.
\end{itemize}

The primitive methods described above have very simple and efficient implementations for the core matrices.  In \cref{table:core} we report the space utilization for each core matrix, along with the time complexity for one of the most important primitive methods (matrix-vector product).  We compare the complexity of the core implicit matrices with their standard dense and sparse representations, showing significant reductions in space usage -- up to a factor $n^2$.  We observe a similar reduction in time complexity of matrix-vector products by up to a factor of $n$.



%

\begin{table}
\caption{
Comparison of core implicit matrices to their corresponding sparse and dense representations, in terms of space usage and time complexity of a matrix-vector product. For sparse and dense matrices, the time complexity is the same as the space complexity.
} \label{table:core}
\centering
\begin{tabular}{|c|c|c||c|c|}
\hline
 & \multicolumn{2}{c||}{Implicit} & Dense & Sparse \\
Core Matrix       & Space Usage & Time (mat-vec) & Space/Time & Space/Time \\ \hline
$Identity$     & $O(1)$        & $O(n)$          & $O(n^2)$ & $O(n)$ \\
$Ones$         & $O(1)$        & $O(m+n)$          & $O(m n)$ & $O(m n)$ \\
$Prefix$       & $O(1)$        & $O(n)$          & $O(n^2)$ & $O(n^2)$ \\
$Suffix$       & $O(1)$        & $O(n)$          & $O(n^2)$ & $O(n^2)$ \\
$Wavelet$      & $O(1)$        & $O(n \log{n})$  & $O(n^2)$ & $O(n \log{n})$ \\\hline
\end{tabular}
\vspace{2ex}
\end{table}

\subsubsection*{Composing matrices}

Core matrices and arbitrary sparse or dense matrices can be combined using a Union, Product (including with a constant), and Kronecker Product to form new matrices that are implicit (or partially implicit).

If matrix $\Q_1$ and $\Q_2$ each represent queries, then Union$(\Q_1,\Q_2)$ is a matrix that represents the union of the queries of $\Q_1$ and $\Q_2$.  It is useful for building complex workloads and measurement matrices, and it also important in plans to bring together all measured queries for global inference.  Product is less frequently needed, but is used for multiplying partition matrices with workload and measurement matrices.

Kronecker product is especially useful for constructing workload and measurement matrices over multi-dimensional domains.  Suppose our input is a relation $R(A,B)$, we vectorize its projection of $\pi_A(R)$ to get data vector $\x_A$, and we define a set of queries of interest as matrix $\Q_A$.  If we similarly form a matrix of queries $\Q_B$ over the vectorization of $\pi_B(R)$ then Kronecker$(\Q_A, \Q_B)$ (denoted $\Q_A \otimes \Q_B$ in matrix equations) is a matrix that encodes a new set of queries over both attributes $A$ and $B$ and it contains $q_{a_i} \wedge q_{b_j}$ for each $q_{a_i}$ in $\Q_A$ and each $q_{b_j}$ in $\Q_B$, i.e. it contains the conjunctive combination of all pairs of queries drawn from $\Q_A$ and $\Q_B$.

The formal definition of the Kronecker product is:

\begin{definition} \label{def:kron}
The Kronecker product $\mathbf{A} \otimes \mathbf{B}$ between a $ m_A \times n_A $ matrix $\mathbf{A}$ and a $m_B \times n_B$ matrix $\mathbf{B}$ is a $m_A m_B \times n_A n_B$ matrix defined as:
$$ \mathbf{A} \otimes \mathbf{B} =
\begin{bmatrix}
a_{11} \B & \dots & a_{1 n_A} \B \\
\vdots  & \ddots & \vdots \\
a_{m_A 1} \B & \dots & a_{m_A n_A} \B \\
\end{bmatrix} $$
\end{definition}

Kronecker products were first used within the context of a privacy mechanism by McKenna et al. \cite{mckenna2018optimizing}.

We present these as binary operations on implicit matrices, but because they are associative, they can also be applied to a collection of $k$ sub-matrices.  For example for $k=3$ we may write $Union(\A,\B,\mathbf{C})$ as shorthand notation for $Union(\A, Union(\B,\mathbf{C}))$. 

\begin{example} 
\label{ex:implicit}
Suppose our input relation is $R(age, income, marital \text{-} status)$, where age and income are discretized into a 100 bins, and marital-status is a categorical attribute with 7 possible values, resulting in a data vector of size 70000.  We want to accurately answer range queries on age and income, broken down by various marital statuses.  Thus, we may construct the following workload using tools from above:
\begin{align*}
\W = Kronecker(
&Prefix, \\
&Prefix, \\
&Union(Total, Identity, Dense)
\end{align*}
where $Dense$ is a $ 2 \times 7 $ query matrix with two queries that aggregate over the marital status attribute into two groups: ``married'' and ``unmarried''.  Using \cref{table: composed}, we see that the only storage required to represent $\W$ is the $ 2 \times 7 $ $Dense$ matrix, and metadata for the $Prefix$, $Total$, and $Identity$ matrices.  In contrast, the sparse and dense representation of $\W$ would require about 8 GB and 56 GB respectively.
\end{example}

\subsubsection*{Supporting the primitive methods}
%
Core, sparse, and dense matrices have native support for the primitive methods discussed previously.  When a primitive method is invoked for an $EMatrix$ that results from one or more of the combining operations, the work is delegated to the constituent sub-matrices, and thus the matrices formed by composition inherit the performance characteristics of the sub-matrices.  In particular, 
the key primitive methods can be implemented efficiently on matrices formed from unions, products, and Kronecker products.


The key performance characteristics of composed matrices are summarized in \cref{table: composed}.


\begin{table}[]
\centering
\caption{Space and time complexity of composed matrices, in terms of the complexity of sub-matrices.} \label{table: composed}
\begin{tabular}{|c|c|c|c|}
\hline
Matrix       & Operation & Space Usage & Time (mat-vec) \\ \hline
$Union(\A,\B)$      & $\begin{bmatrix} \A \\ \B \end{bmatrix}$ & Space($\A$) + Space($\B$)   & Time($\A$) + Time($\B$) \\
$Product(\A,\B)$    & $ \A \B $ & Space($\A$) + Space($\B$)   & Time($\A$) + Time($\B$)  \rule{0pt}{2.5ex}\\
$Kronecker(\A,\B)$  & $ \A \otimes \B $ & Space($\A$) + Space($\B$)   & $n_B$ Time($\A$) + $m_A$ Time($\B$) \rule{0pt}{2.5ex} \\   \hline
\end{tabular}
\vspace{1.5ex}
\end{table}


%
%
%
%
%
%
%
%
%
%
\subsection{Matrix constructions for \sys operators} \label{sec:range_queries}

Using the generalized matrix construction described above, we can re-implement many of the existing query selection and partition operators currently in \sys. We will show in \cref{sec:experiments} that the use of implicit matrices leads to significant improvements in efficiency and scalability, as well as the reduction of space consumption.

\subsubsection*{Query selection operators based on range queries}
A notable class of query matrices that we can efficiently represent using the tools from above is an arbitrary collection of range queries. This type of workload has been extensively studied in the literature and many of the query selection operators in \sys are specially-designed sets of range queries including H2, Hb, QuadTree, UniformGrid, and AdaptiveGrid.

Recall that a single range query over a 1-dimensional domain can be specified by a pair of indices $ (i,j) $, and a workload of range queries can be represented as a list of these pairs.  This suggests we can store any range query workload using only $O(m)$ space where $m$ is the number of queries.  Furthermore, one way to evaluate matrix-vector products is by iterating through each query one-by-one and evaluating $ \sum_{k=i}^j x_k $.  The time complexity of this approach is $O(mn)$ in the worst case, which is equivalent to the sparse and dense representations.

Our general matrix construction allows us to do even better by exploiting the fact that any range query can be expressed as the difference of two prefix queries.  Thus the matrix can be represented as $Product(Sparse, Prefix)$ where $Sparse $ is a $m \times n$ sparse matrix with two non-zero entries per row.  An illustrating example is shown below:

\begin{example}[Range Queries] \label{ex:range}
A collection of four range queries over a domain of size five, represented implicitly as the product of a $Sparse$ matrix and the $Prefix$ matrix (displayed here in dense form for illustration purposes):
$$\begin{bmatrix}
0 & 1 & 1 & 1 & 0 \\
0 & 0 & 0 & 1 & 1 \\
1 & 1 & 1 & 1 & 0 \\
0 & 1 & 0 & 0 & 0
\end{bmatrix} =
\begin{bmatrix}
-1 & 0 & 0 & 1 & 0 \\
0 & 0 & -1 & 0 & 1 \\
0 & 0 & 0 & 1 & 0 \\
-1 & 1 & 0 & 0 & 0 \\
\end{bmatrix}
\begin{bmatrix}
1 & 0 & 0 & 0 & 0 \\
1 & 1 & 0 & 0 & 0 \\
1 & 1 & 1 & 0 & 0 \\
1 & 1 & 1 & 1 & 0 \\
1 & 1 & 1 & 1 & 1 \\
\end{bmatrix} $$
\end{example}
Using this construction, we can evaluate matrix-vector products in $O(n+m)$ time, which is a substantial improvement over the other representations.  
The range query construction can be naturally extended to multi-dimensional domains by replacing $ Prefix $ with $ Kronecker(Prefix, \dots, Prefix) $ and replacing $Sparse$ with a sparse matrix with up to $ 2^d $ nonzero entries per row, where $d$ is the number of dimensions of the domain.

A special case of this range query construction is hierarchical and grid-based matrices used by H2, Hb, and QuadTree.  These matrices always have an Identity matrix, and while they can be represented using the above construction, it is more efficient to represent them in a slightly different way as $ Union(Identity, Product(Sparse, Prefix)) $, which is the representation used in our empirical evaluation.

Note that even though $Product$ does not natively support $abs$ and $sqr$, for the case of range queries, or more generally any matrix with binary values, $abs$ and $sqr$ are simple no-ops.

\subsubsection*{Representing marginals}

A common task for multi-dimensional data analysis is computing the marginals of a dataset.  Marginals may be used both as part of workloads and measurement matrices.  They can be efficiently represented using the tools from above, as demonstrated in Example \ref{ex:marginals}.


\begin{example}[Marginals] \label{ex:marginals}
Any marginal can be represented as a Kronecker product of $Identity$ and $Total$ building blocks.  For example, the two-way marginal that sums out the second attribute can be encoded as:
$$ \W_{13} = Kronecker(Identity, Total, Identity) $$
Further, an arbitrary collection of marginals can be encoded as $Union$ of these Kronecker products.  All 2-way marginals is:
\begin{align*}
\W_{2way} = Union(
&Kronecker(Identity, Identity, Total), \\
&Kronecker(Identity, Total, Identity),  \\
&Kronecker(Total, Identity, Identity))
\end{align*}
\end{example}

\subsubsection*{Partition operators}

The matrices used by partition operators are represented simply as $Sparse$ matrices.  While an implicit definition is possible, it would not offer any improvement in space or time over the sparse representation.



%

\subsection{Implementing inference} \label{sec:improved_inference}

Inference is a fundamental operator that can improve error with no cost to privacy and, accordingly, we saw that it appeared in virtually every algorithm re-implemented in \sys (as shown in \cref{fig:plan_signatures}).  But inference can be a costly operation.  Recall that the input to inference is a measurement matrix, denoted by $\Q$, containing $m$ queries defined over a data vector of size $n$, and the list of noisy answers $\y$.  The least squares solution (Eq.~(\ref{ls})) is given by the solution to the normal equations $ \Q^T \Q \hat{\x}=\Q^T \y $.  Assuming $ \Q^T \Q $ is invertible, then the solution is unique and can be expressed as $\hat{\x} = (\Q^T \Q)^{-1} \Q^T \y$.  Often explicit matrix inversion is avoided, in favor of suitable factorizations of $\Q$ (e.g., QR or SVD).  However, the time complexity of such ``direct'' methods is still generally cubic in the domain size when $ \qsize = O(n)$.  In practice we have found that the runtime of such direct methods is unacceptable when $n$ is greater than about $ 5000 $.

 Algorithms in prior work~\cite{xiao2010differential,hay2010boosting,Qardaji13Understanding,qardaji2013differentially} have used least squares inference on large domains by restricting the selection of queries, namely to those representing a set of hierarchical queries.  This allows for inference in time linear in the domain size, avoiding the explicit matrix representation of the queries.  We avoid this approach in \sys because it means that a custom inference method may be required for each query selection operation, and because it limits the measurement sets that can be used.  In addition, hierarchical methods only work for least squares but not least squares with non-negativity constraints.



An alternative approach to least squares inference is to use an iterative gradient-based method, which solves the normal equations by repeatedly computing matrix-vector products $ \Q \mathbf{v} $ and $ \Q^T \mathbf{v} $ until convergence.  The time complexity of these methods is $O(k n^2)$ for dense matrix representations where $k$ is the number of iterations.  In experiments we use a well-known iterative method, LSMR \cite{Fong11LSMR:}.  Empirically, we observe LMSR to converge in far fewer than $n$ iterations when $\Q$ is well-conditioned, which is the case as long as the queries are not taken with vastly different noise scales, and thus we expect $ k << n $.

In the original version of this work \cite{Zhang:2018:EFD:3183713.3196921}, we demonstrated a significant performance benefit from the combination of iterative solution methods with sparse matrix representations.  This benefit is amplified when the underlying matrix representation is implicit.  Letting $Time(\Q)$ denote the time complexity of evaluating a matrix-vector product with $\Q$, the time complexity of least squares inference is $ O(k \cdot Time(\Q)) $ where $k$ is the number of iterations, as before.  As shown in \cref{table:core} and \cref{table: composed}, $Time(\Q)$ is often $O(n)$, resulting in a very favorable $O(k n)$ time complexity for inference.

Iterative approaches, using implicit matrices, are well-suited to the other inference methods in \sys: least squares with non-negativity constraints (Eq.~(\ref{nnls})) and multiplicative weights.  For the former, we use the limited memory BFGS algorithm with bound constraints \cite{byrd1995limited}.  The time complexity of this algorithm is the same as LSMR, although the number of iterations needed for convergence may be different, and there is a constant factor overhead for storing the low-rank approximation to the inverse Hessian matrix.  The multiplicative weights inference algorithm is defined in an iterative manner and requires the same primitive methods as ordinary least squares and non-negative least squares.




The combination of our general matrix construction techniques with iterative inference result in flexible inference capabilities for plan authors.  With relative freedom, they can construct measurement matrices, or combine measurements from parts of a plan, and apply a single generic inference operator, which will run efficiently.  In Sec.~\ref{sec:experiments} we show that using iterative least squares on implicitly represented matrices, we can scale inference to domains consisting of hundreds of millions of cells while staying within modest runtime bounds, well beyond what is possible with sparse or dense representations.

\section{Workload \!-\! based partition selection} \label{sec:wbdr}

In many cases, the goal of a differentially private algorithm (and its corresponding \sys plan) is to answer a given workload of queries, $\W$, defined in terms of a data vector $\x$.  We describe next a method for reducing the representation of the $\x$ vector to precisely the elements required to correctly answer the workload queries.  This is a new partition selection operator, called {\em workload-based partition selection}, which can be used as input to a \reduceOp transformation.  

	We define the partition below and prove that, under reasonable assumptions, using such domain reduction can never hurt accuracy.  We provide an algorithm for computing the partition, which can be executed using implicit workload representations.  Later, in \cref{sec:sub:wbdr_exp}, we will show empirically that using this partition in plans can offer significant improvement in both runtime and error.  


\subsection{The workload-based partition and its properties} \label{sec:sub:wkld}

For a workload $\W$ of linear queries described on data vector $\x$, it is often possible to define a reduction of $\x$, to a smaller $\x'$, and appropriately transform the workload to $\W'$, so that all workload query answers are preserved, i.e. $\W\x = \W'\x'$.  Intuitively, such a reduction is possible when a set of elements of $\x$ is not distinguished by the workload: each linear query in the workload either ignores it, or treats it in precisely the same way.  In that case, that portion of the domain need not be represented by multiple cells, but instead by a single cell in a reduced data vector.  It is in this sense that the reduction is lossless with respect to the workload. Following this intuition, the domain reduction can be computed from the matrix representation $\W$ of the workload by finding groups of identical columns: elements of these groups will be merged in $\W$ to get $\W'$ while the corresponding cells in $\x$ are summed.
\begin{example} \label{ex:reduction}
Consider a table with schema Census(age, sex, salary). If the workload consists of queries Q1$(salary \leq 100K, sex=M)$ and Q2$(salary > 100K, sex=F)$ the workload only requires a data vector consisting of 2 cells. If the workload consists of all 1-way marginals then no workload-based data reduction is possible.
\end{example}
Note that calculating this \hl{partition} only requires knowledge of the workload and is therefore done in the unprotected client space (and does not consume the privacy budget).  The \hl{partition} is then input to a \reduceOp transformation operator carried out by the protected kernel and its stability is one.

The new \hl{workload-based partition selection operator} can be formalized in terms of a linear matrix operator, as follows:

\begin{definition}[Workload-based \hl{partition selection}] \label{def:P}
	Let $\vect{w}_1, \dots, \vect{w}_n$ denote the columns of $\W$ and let $ \vect{u}_1, \dots, \vect{u}_{\psize} $ denote those that are unique.  For $h(\vect{u}) = \{ j \mid \vect{w}_j = \vect{u} \} $, define the transformation matrix $\P \in \mathbb{R}^{\psize \times n} $ to have $P_{ij}=1$ if $j \in h(\vect{u}_i)$ and $ P_{ij} = 0 $ otherwise. The reverse transformation is the pseudo-inverse $\P^{+} \in \mathbb{R}^{n \times \psize}$.
\end{definition}
The matrix $\P$ defines a \hl{partition} of the data, which can be passed to \reduceOp to transform the data vector, and $\P^+$ can be used to transform the workload accordingly.
When $\P$ is passed to \reduceOp, the operator produces a new data vector $\x' = \P\x$ where $x'_i$ is the sum of entries in $\x$ that belong to $i^{th}$ group of $\P$.  When viewed as an operation on the workload, $\P^+$ merges duplicate columns by taking the row-wise average for each group.  This is formalized as follows:



\begin{proposition}[properties: workload-based reduction]\label{prop:wbdr}
	Given transform matrix $\P$ and its pseudo-inverse $\P^{+}$, the following hold:
	\begin{itemize} \itemsep 0in
		\item $ \x' = \P \x $ is the reduced data vector;
		\item $ \W' = \W \P^+ $ is the workload matrix, represented over $\x'$;
		\item The transformation is lossless: $\W \x = \W' \x' $
	\end{itemize}

\end{proposition}

\begin{proof}
First note that $ \P^+ = \P^T \D^{-1} $ where $\D$ is the $\psize \times \psize$ diagonal matrix with $ D_{ii} = | h(\vect{u}_i) | $ for $h$ defined in Def. \ref{def:P}.  Since $\P$ has linearly independent rows, $ \P^+ = \P^T (\P \P^T)^{-1} $ and $ \P \P^T = \D $ because $ h(\vect{u}_i) $ and $ h(\vect{u}_j) $ are disjoint for $ i \neq j $.  By definition of $\P$, we see that $ x'_i = \sum_{j \in h(\vect{u}_i)} x_j $ for $ 1 \leq i \leq \psize $.  Similarly, the $i^{th}$ column of $\W'$ is given by $\vect{w}'_i = \frac{1}{|h(\vect{u}_i)|} \sum_{j \in h(\vect{u}_i)} \vect{w}_j $.  Since $ \vect{w}_j = \vect{u}_i $ when $ j \in h(\vect{u}_i) $, we have $ \vect{w}'_i = \vect{u}_i $, which shows that $\W'$ is just $\W$ with the duplicate columns removed.  Using these definitions, we show that the transformation is lossless:
$$ \W \x = \sum_{i=1}^n \vect{w}_i x_i = \sum_{i=1}^\psize \vect{u}_i \sum_{j \in h(\vect{u}_i)} x_j = \sum_{i=1}^\psize \vect{w}'_i x'_i = \W' \x' $$
\end{proof}

As noted in Example \ref{ex:reduction}, not all workloads allow for reduction (in some cases, the $\P$ matrix computed above is the identity).  But others may allow a significant reduction, which improves the efficiency of subsequent operators.  Less obvious is that workload-based data reduction would impact accuracy.  In fact, many query selection methods from existing work depend implicitly on the representation of the data in vector form, and these approaches may be improved by domain reduction.  In Sec.~\ref{thm:wbdrthm} we measure the impact of this transform on accuracy and efficiency.



We show next that this reduction does not hurt accuracy: for any selected set of measurement queries, their reduction will provide lower error after transformation.

\begin{restatable}{theorem}{wbdrthm}
	\label{thm:wbdrthm}
	Given a workload $\W$ and data vector $\x$, let $\Q$ be any query matrix that answers $\W$.  Then if $ \vect{q'} = \vect{q} \P^+ $ is a reduced query and $ \Q' = \Q \P^+ $ is the query matrix on the reduced domain, $Error_{\vect{q'}}(\Q') \leq Error_{\vect{q}}(\Q) $ for all $ \vect{q} \in \W$.
\end{restatable}

\begin{proof}
	We use the definition of squared error from~\cite{li2015matrix} which shows that $ \forall \q \in \W $, $ Error_{\q}(\Q) \propto \norm{ \Q }_1^2 \norm{ \q \Q^+ }_2^2 $ as long as $\Q$ supports $\W$.  
Let $\hat{\m}_j$ and $\m_j$ denote the $j^{th}$ column of $\hat{\Q}$ and $ \Q $ respectively.
First we show that $ \norm{ \hat{\Q} }_1 \leq \norm{ \Q }_1 $:

\begin{align*}
\norm{ \hat{\Q} }_1 &= \max_{1 \leq i \leq p} \norm{ \hat{\m}_j } \\
&= \max_{1 \leq i \leq p} \norm{ \frac{1}{|h(\u_j)|} \sum_{j \in h(\u_i)} \m_j } \\
&\leq  \max_{1 \leq i \leq p} \frac{1}{|h(\u_i)|} \sum_{j \in h(\u_i)} \norm{ \m_j }_1 \\
&\leq \max_{1 \leq i \leq p} \max_{j \in h(\u_i)} \norm{ \m_j }_1  = \\
&= \max_{1 \leq i \leq n} \norm{ \m_i }_1  \\
&= \norm{ \Q }_1
\end{align*}

	where $ h(\u) $ and $ \u_i $ are defined in definition \ref{def:P}.  Now we show that $ \norm{ \q \Q^+ }_2 \leq \norm{ \hat{\q} \hat{\Q}^+ }_2 $.  Observe that it is possible to to write $\q$ as a linear combination of the rows of $\Q$ since $\Q$ supports $\W$.  That is, there exists a $\z$ satisfying $ \z \Q = \q $.  In general, there may be infinitely many solutions to this linear system, but $ \z = \q \Q^+ $ is the \emph{minimum-norm} solution~\cite{penrose_1956}.  On the reduced domain, we also know there exists a $ \hat{\z} $ satisfying $ \hat{\z} \hat{\Q} = \hat{\q} $, or equivalently $ \hat{\z} \Q \P^+ = \q \P^+ $.   By making the substitution $ \z \Q = \q $, it's easy to see that $ \hat{\z} = \z $ is one solution to this linear system.  The minimum norm solution to this linear system is $ \hat{\z} = \hat{\q} \hat{\Q}^+ $, which implies $\norm{ \hat{\z} }_2 \leq \norm{ \z }_2 $.  This shows that $ \norm{ \hat{\q} \hat{\Q}^+ }_2 \leq \norm{ \q \Q^+ }_2 $, and it immediately follows that $ Error_{\hat{\q}}(\hat{\Q}) \leq Error_{\q}(\Q) $ as desired.
\end{proof}

\subsection{Computing the partition}\label{appendix:wbdr}

 The computation of the \hl{partition} $\P$ in Def. \ref{def:P} is conceptually straightforward: it simply requires finding the unique columns of $\W$ and grouping them.  Finding the unique columns of $\W$ exactly by inspecting the entries of $\W$ requires an explicit matrix in dense form, or materializing an implicit matrix.  Algorithm \ref{alg:reduce} provides an efficient method for finding the column groupings that does not require a explicit matrix representation, relying instead only on the primitive methods of transpose and matrix-vector product.  This approach is highly scalable and is compatible with the implicit matrix representations of the workload discussed in section~\ref{sec:implementation}.

\begin{algorithm}
	\caption{\label{alg:reduce} An algorithm for workload-based data reduction}
	
		\begin{algorithmic}[1]
			\Procedure{Compute reduction matrix}{$W$}
			\State \textbf{Input:} $ \wsize \times n$ matrix $\W$
			\State \textbf{Output:} $ \psize \times n $ matrix $\P$ where $ \psize \leq n $
			\State set $ \vect{v} $ = vector of $\wsize$ samples from Uniform$(0,1) $ \Comment $ 1 \times \wsize $
			\State compute $\vect{h}$ = $ \W^T \vect{v} $ \Comment $ 1 \times n \;$
			\State let $G = g_1, \dots, g_{\psize}$ be groups of common values in $\vect{h}$
			\State initialize matrix $\P$ with zeros \Comment $ \psize \times n \,$
			\For {$g_i$ in $G$}
			\State set row $i$ of $\P$ to $1$ in each position of $g_i$
			\EndFor
			\State return $\P$
			\EndProcedure
	\end{algorithmic}
\end{algorithm}


By grouping the elements of $\vect{h}$ (line 6) we recover the column groupings of $\W$, because if $ \vect{w}_i = \vect{w}_j $ then $ h_i = h_j $ and if $ \vect{w}_i \neq \vect{w}_j $ then $ P(h_i = h_j) = 0 $ since $h_i$ and $h_j$ are continuous random variables.  While algorithm \cref{alg:reduce} is a randomized algorithm, it returns the correct result almost surely.  The probability of incorrectly grouping two different columns is approximately $10^{-16}$ with a 64-bit floating point representation, but if needed we can repeat the procedure $k$ times until the probability of failure ($\sim 10^{-16k}$) is vanishingly small.

\section{Case studies: \sys in action} \label{sec:case_studies}

In this section we put \sys into action by developing new algorithms.  First, we show that it is easy to re-design and improve existing algorithms by combining operators in new ways. In particular, we develop a variant of the MWEM algorithm with significantly improved accuracy.  Then we use \sys to tackle two practical use-cases, constructing new plans which offer state-of-the-art accuracy.  We evaluate all of the proposed plan in \cref{sec:experiments}. 

\subsection{Recombination of operators to improve MWEM} \label{sec:mwem_variants}

Using \sys, we design new variants of the well-known {\em Multiplicative Weights Exponential Mechanism} (MWEM) \cite{hardt2012a-simple} algorithm.  MWEM repeatedly derives the worst-approximated workload query with respect to its current estimate of the data, then measures the selected query, and uses the multiplicative weights update rule to refine its estimate, often along with any past measurements taken.  This repeats a number of times, determined by an input parameter.

    When viewed as a plan in \sys, a deficiency of MWEM becomes apparent.  Its query selection operator selects a \emph{single} query to measure whereas most query selection operators select a set of queries such that the queries in the set measure disjoint partitions of the data.  By the parallel composition property of differential privacy, measuring the entire set has the same privacy cost as asking any single query from the set.  This means that MWEM could be measuring more than a single query per round (with no additional consumption of the privacy budget).
To exploit this opportunity, we designed an augmented query selection operator that adds to the worst-approximated query by attempting to build a binary hierarchical set of queries over the rounds of the algorithm.  In round one, it adds any unit length queries that do not intersect with the selected query.  In round two, it adds length two queries, and so on.

Adding more measurements to MWEM has an undesirable side effect on runtime, however.  Because it measures a much larger number of queries across rounds of the algorithm and the runtime of multiplicative weights inference scales with the number of measured queries, inference can be considerably slower.  Thus, we also use replace it with a version of least-squares with a non-negativity constraint (NNLS) and incorporate a high-confidence estimate of the total which is assumed by MWEM.


In total, we consider three MWEM variants: an alternative query selection operator (\planref{\planmwemb}), an alternative inference operator (\planref{\planmwemc}), and the addition of both alternative operators (\planref{\planmwemd}). These are shown in \cref{fig:plan_signatures} and evaluated in \cref{sec:experiments}.

\subsection{Census case-study} \label{sec:census_plans}

The U.S. Census Bureau collects data about U.S. citizens and releases a wide variety of tabulations describing the demographic properties of individuals.
We consider a subset of the (publicly released) March 2000 Current Population Survey.  The data report on 49,436 heads-of-household describing their income, age (in years), race, marital status, and gender.  We divide Income into $5000$ uniform ranges from $(0,750000)$, age in 5 uniform ranges from $(0,100)$, and there are 7, 4 and 2 possible values for status, race and gender.

We author differentially private plans for answering a workload of queries  similar to Census tabulations.  This is challenging because the data domain is large and involves multiple dimensions.  The workloads we consider are: (a) the Identity workload (or counts on the full domain of 1.4M cells), (b) a workload of all 2-way marginals (age $\times$ gender, race $\times$ status, and so on),  and (c) a workload suggested by U.S. Census Bureau staff: Prefix(Income) which consists of all counting queries of the form (income $\in (0, i_{high}$), age=$a$, marital=$m$, race=$r$, gender=$g$) where $(0, i_{high})$ is an income range, and $a, m, r, g$ may be values from their resp. domains, or $<\!any\!>$.




There are few existing algorithms suitable for this task.  We were unable to run the DAWA \cite{Li14Data-} algorithm directly on such a large domain. In addition, it was designed for 1d- and 2d- inputs.  One of the few algorithms designed to scale to high dimensions is PrivBayes \cite{Zhang2014}.  While not a workload-adaptive algorithm, PrivBayes generates synthetic data which can support the census workloads above. We use PrivBayes as a baseline and we use \sys to construct three new plans composed of operators in our library. The proposed plans are: \algoname{Hb-Striped} (\planref{\planhbstripe}), \algoname{Dawa-Striped} (\planref{\plandawastripe}), and \algoname{PrivBayesLS} (\planref{\planprivbayesls}).  The first two ``striped'' plans showcase the ability to adapt lower dimensional techniques to a higher dimensional problem avoiding  scalability issues. The third plan considers improving on  PrivBayes  by changing its inference step.

Both \algoname{Hb-Striped} and \algoname{Dawa-Striped} use the same plan structure: first they partition the full domain, then they execute subplans to select measurements for each partition, and lastly, given the measurement answers, they perform inference on the full domain and answer the workload queries.  The partitioning of the initial step is done as follows: given a high dimensional dataset with $N$ attributes and an attribute $A$ of that domain, our partitions are parallel ``stripes'' of that domain for each fixed value of the rest of the $N-1$ attributes, so that the measurements are essentially the one-dimensional histograms resulting from each stripe.  In the case of \algoname{Hb-Striped} (fully described in \cref{alg:Hb-striped}), the subplan executed on each partition is the \algoname{Hb} algorithm \cite{Qardaji13Understanding}, which builds an optimized hierarchical set of queries, while in the case of the \algoname{Dawa-Striped} the subplan is \algoname{DAWA} algorithm \cite{Li14Data-}.  Note that while the data-independent nature of the \algoname{Hb} subplan means that all the measurements from each stripe are the same, that is not the case with \algoname{Dawa}, which potentially selects different measurement queries for each stripe, depending on the local vector it sees.  For our experiments, the attribute chosen was Income, and for \algoname{Dawa-Striped} we set the \algoname{DAWA} parameter $\rho$ to $0.25$. 
\begin{algorithm}[h]
\caption{\algoname{Hb-Striped}}\label{alg:Hb-striped}
{\small
\begin{algorithmic}[1]
\State $D \gets$  \algoname{Protected}(source\_uri) \Comment{Init}
\State $\mathbf{x} \gets$ \algoname{T-Vectorize}($D$)		\Comment{Transform}
\State $\mathbf{R} \gets$ StripePartition(Attr)	\Comment{Partition Selection}
\State $\mathbf{x}_R \gets$ \algoname{\partitionOp}($\mathbf{x}$, $\mathbf{R}$)
\State $\mathbf{M} \gets \emptyset$
\State $\mathbf{y} \gets \emptyset$
\For {$\mathbf{x'} \in \mathbf{x}_R$}
	\State $\mathbf{M} \gets$ $\mathbf{M} \cup$ \algoname{Hb}($\mathbf{x'}$) 		\Comment{Query Selection}
	\State $\mathbf{y} \gets$ $\mathbf{y} \cup$ \algoname{VecLaplace}($\mathbf{x'}$, $\mathbf{M}$, $\epsilon$) 	\Comment{Query}
\EndFor
\State $\hat{\mathbf{x}}  \gets$ LS($\mathbf{M}, \mathbf{y}$)
\State \Return  $\hat{\mathbf{x}}$	\Comment{Output}
\end{algorithmic}
}
\end{algorithm}

Selection in \op{HB-Striped}'s subplans are data-independent, unlike in \op{DAWA-Striped}, so the exact same set of measurements will be selected on each partition. As introduced in \cref{sec:sub:matrix_construction}, this set of measurements can be represented compactly as a Kronecker product. So we introduce a new selection operator \op{Stripe(attr)} where a global measurement is composed by constructing the Kronecker product of HB measurements on the stripe dimension and Identity matrices on other dimensions. \op{HB-Striped\_kron} (\planref \planhbstripekron) is a sequence starting with the new \op{SS} selection operator, followed by Laplace measurement and LS inference.  The complete plan is shown in \cref{alg:Hb-striped-kron}.
This non-iterative alternative implementation is more efficient and we compare the efficiency and scalability of the two implementations in \cref{sec:sub:plan_efficiency}.

\begin{algorithm}[h]
\caption{\algoname{Hb-Striped\_kron}}\label{alg:Hb-striped-kron}
{\small
\begin{algorithmic}[1]
\State $D \gets$  \algoname{Protected}(source\_uri) \Comment{Init}
\State $\mathbf{x} \gets$ \algoname{T-Vectorize}($D$)		\Comment{Transform}

\State $\mathbf{M} \gets$ \algoname{StripeSelect}(Attr) 		\Comment{Query Selection}
\State $\mathbf{y} \gets$ \algoname{VecLaplace}($\mathbf{x}$, $\mathbf{M}$, $\epsilon_3$) 	\Comment{Query}

\State $\hat{\mathbf{x}}  \gets$ LS($\mathbf{M}, \mathbf{y}$)
\State \Return  $\hat{\mathbf{x}}$	\Comment{Output}
\end{algorithmic}
}
\end{algorithm}

%
%
%

Our final plan is a variant of \algoname{PrivBayes} in which we replace the original inference method with least squares, retaining the original \algoname{PrivBayes} query selection and query steps.  We call this algorithm \algoname{PrivBayesLS} and it's fully described in \cref{alg:PrivBayesLS}.

\begin{algorithm}[h]
\caption{ \algoname{PrivBayesLS}}\label{alg:PrivBayesLS}
{\small
\begin{algorithmic}[1]
\State $D \gets$  \algoname{Protected}(source\_uri) \Comment{Init}
\State $\mathbf{x} \gets$ \algoname{T-Vectorize}($D$)		\Comment{Transform}
\State $\mathbf{M} \gets$ \algoname{PBSelect}($\mathbf{x}, \epsilon_2$) 		\Comment{Query Selection}
\State $\mathbf{y} \gets$ \algoname{VecLaplace}($\mathbf{x}$, $\mathbf{M}$, $\epsilon_3$) 	\Comment{Query}
\State $\hat{\mathbf{x}}$ $\gets$ LS($\mathbf{M}$, $\mathbf{y}$)			\Comment{Inference}
\State
\Return $\mathbf{W}\cdot \hat{\mathbf{x}}$	\Comment{Output}
\end{algorithmic}
}
\end{algorithm}

We evaluate the error incurred by these plans in \cref{sec:census:exp}, and show that the best of our plans outperforms the state-of-the-art \algoname{PrivBayes} by at least 10$\times$ in terms of error.

\subsection{Naive Bayes case-study}
We also demonstrate how \sys can be used for constructing a  Naive Bayes classifier. To learn a NaiveBayes classifier that predicts a binary label attribute $Y$ using predictor variables $(X_1, \ldots, X_k)$ requires computing 2k+1 1d histograms: a histogram on $Y$, histogram on each $X_i$ conditioned on each value on $Y$. We design \sys plans to compute this workload of 2k+1 histograms, and  use them to fit the classifier under the Multinomial statistical model \cite{kotsogiannis2017pythia}.



We develop two new plans and compare them to two plans that correspond to algorithms considered in prior work. \algoname{Workload} represents the 2k+1 histograms as a matrix, and uses \op{Vector Laplace} to estimate the histogram counts. This corresponds to a technique proposed in the literature \cite{cormode2011personal}.   The other baseline is \algoname{Identity} (Plan \planidentity): it estimates all  point queries in the contingency table defined by the attributes, adds noise to it,  and  marginalizes the noisy contingency table to compute the histograms.

The first new plan is \algoname{WorkloadLS} which runs \algoname{Workload} followed by a least squares inference operator, which for this specific workload would make all  histograms have consistent totals. Our second plan is called \algoname{SelectLS} (fully described in \cref{alg:SelectLS}) and  selects a different algorithm (subplan) for estimating each of the histograms.
\begin{algorithm}[h]
\caption{\algoname{SelectLS}}\label{alg:SelectLS}
{\small
\begin{algorithmic}[1]
\State $D \gets$  \algoname{Protected}(source\_uri) \Comment{Init}
\State $\mathbf{x} \gets$ \algoname{T-Vectorize}($D$)		\Comment{Transform}
\State $\mathbf{R} \gets$ \algoname{MargReduction}($x$, Att)	\Comment{Partition Selection}

\State $\mathbf{M} \gets \emptyset$, $\mathbf{y} \gets \emptyset$

\For {$i = 1:k$} \Comment{Iterate over Dimensions}
	\State $\mathbf{x'} \gets $ \algoname{\reduceOp}($\mathbf{x}, \mathbf{R}_i$)
	\If {$\mathbf{DomainSize}_i > 80$}
		\State $\mathbf{R'} \gets$ \algoname{RDawa}  ($\mathbf{x'}, \epsilon_1$/k) \Comment{Partition Selection}
		\State $\mathbf{x'}_R \gets $ \algoname{\reduceOp}($\mathbf{x}$, $\mathbf{R'}$)
		\State $\mathbf{M} \gets $ \algoname{GreedyH}($\mathbf{x'}_R$) \Comment{Query Selection}
		\State $\mathbf{y} \gets$ $\mathbf{y} \cup$ \algoname{VecLaplace}($\mathbf{x'}_R$, $\mathbf{M}$, $\epsilon_2$/k) 	\Comment{Query}
	\Else
		\State $\mathbf{M} \gets$ $\mathbf{M} \cup$ \algoname{Identity}($\mathbf{x'}$) 		\Comment{Query Selection}
		\State $\mathbf{y} \gets$ $\mathbf{y} \cup$ \algoname{VecLaplace}($\mathbf{x'}$, $M$, $\epsilon$/k) 	\Comment{Query}
	\EndIf
	\State $\mathbf{x} \gets $ \algoname{\reduceOp}($\mathbf{x'}, R_i$) \Comment{Domain Expansion}
\EndFor
\State $\hat{\mathbf{x}}  \gets$ LS($\mathbf{M}, \mathbf{y}$)
\State \Return  $\hat{x}$	\Comment{Output}
\end{algorithmic}
}
\end{algorithm}
\algoname{SelectLS} first runs 2k+1 domain reductions to compute 2k+1 vectors, one for each histogram. Then, for each vector, \algoname{SelectLS} uses a conditional statement to select between two subplans: if the vector size is less than $80$, \algoname{Identity} is chosen, else a subplan that runs DAWA  \hl{partition} selection followed by \algoname{Identity} is chosen.
We combine the answers from all subplans and use least squares inference jointly on all measurements. The inputs to the inference operator are the noisy answers and the workload of effective queries on the full domain.  In \cref{sec:nb:exp} we show that our new plans not only outperform existing plans, but also approach the accuracy of the non-private classifier in some cases.

\section{Experimental evaluation} \label{sec:experiments}

Our prototype implementation of \sys, including all algorithms and variants used below, consists of $7.9k$ lines of code: $25\%$ is the framework itself, $46\%$ consist of operator implementations, $14\%$ consist of definitions of plans used in our experiments  and the remaining $15\%$ are tests and examples provided for the users.

In this section, we first report the results of using \sys in the case studies of \cref{sec:case_studies}. Then compare efficiency and scalability of plans and inference operators with different implementation choices. Finally, we evaluate the impact of workload-based domain reduction introduced in \cref{sec:sub:wkld}.

\subsection{Case studies}



\subsubsection{MWEM: improved query selection \& inference} \label{sec:sub:mwem}

We evaluate the three new plans described in \cref{sec:mwem_variants} which were variants of \sys plan for the MWEM~\cite{hardt2012a-simple} algorithm.  Recall that the variants were achieved by replacing key operators in the MWEM plan.  These algorithms are data-dependent algorithms so we evaluate them over a diverse collection of 10 datasets taken from DPBench \cite{Hay16Principled}.  The results are shown in \cref{tbl:mwem_variants}.

\begin{table}[h]
\centering
\caption{\label{tbl:mwem_variants} For three new algorithms, (b), (c), and (d), the multiplicative factors by which error is improved, presented as (min, mean, max) over datasets.  For runtime, the mean is shown, normalized to the runtime of standard MWEM. (1D, n=4096, W=RandomRange(1000), $\epsilon=0.1$)}
\includegraphics[width=.8\columnwidth]{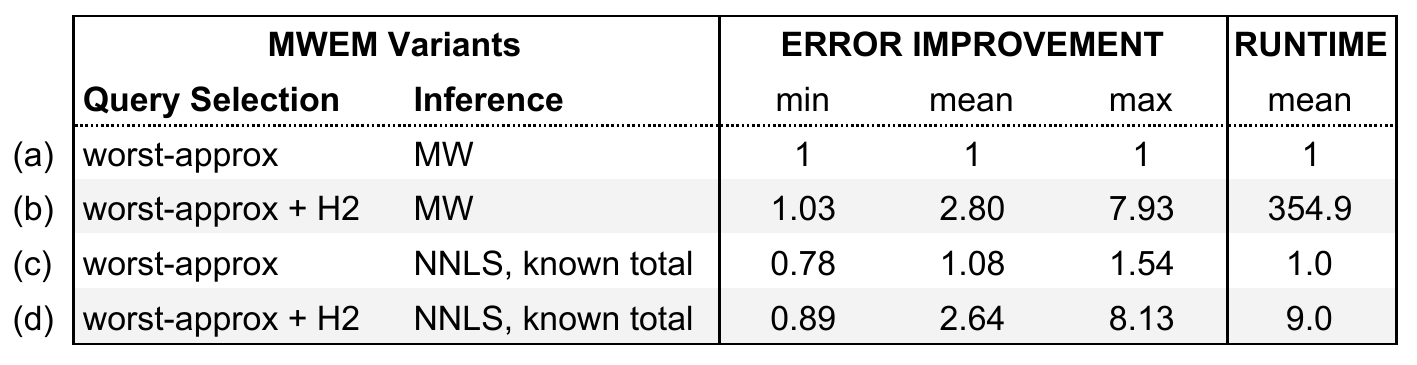}
\end{table}
%
%
The performance of the first variant, line (b), shows that the augementing query selection with H2 can significantly improve error: by a factor of 2.8 on average (over various input datasets) and by as much as a factor of 7.9.
(Error and runtime measures are normalized to the values for the original MWEM; min/mean/max error values represent variation across datasets.)
Unfortunately, this operator substitution has a considerable impact on performance: the added queries slow down by a factor of more than 300.  But combining augmented query selection with NNLS inference, line (d), reduces runtime significantly: it is still slower than the original MWEM algorithm, but by only a factor of 9.  Using the original MWEM query selection with NNLS inference, line (c), has largely equivalent error and runtime to the original MWEM.  Thus, the performance gains of NNLS inference over MW appear to be most pronounced when the number of measured queries is large.





\subsubsection{Census data analysis}\label{sec:census:exp}


In this section we compare the \sys plans proposed in \cref{sec:census_plans}, measuring their effectiveness in computing workloads inspired by Census tabulations.  We compare our three new plans \algoname{PrivBayesLS},  \algoname{Hb-Striped}, and \algoname{DAWA-Striped} with baseline algorithms \algoname{Identity} (\planref{1} in \cref{fig:plan_signatures}) and \algoname{PrivBayes}, our \sys reimplementation of a state-of-the-art algorithm for high dimensional data \cite{Zhang2014}.

\eat{A prime competitor for this task is PrivBayes \cite{Zhang2014}, a data-dependent algorithm suitable for multi-dimensional datasets.  Given an input dataset $D$ it synthesizes a private version of it, $\tilde{D}$. PrivBayes works in three steps: first it spends $\epsilon_1 = \epsilon /2 $ of its privacy budget to learn, using the Exponential Mechanism \cite{mcsherry2007mechanism,Dwork14Algorithmic}, a Bayesian network that best describes the underlying dependencies between the attributes;
second it uses the Laplace mechanism to estimate the conditional probability distributions for each of the nodes of the network by using a total of  $\epsilon_2 = \epsilon / 2$ privacy budget; and third, it uses the Bayes network to estimate the  full joint distribution of the data
and samples $N$ tuples from the noisy distribution to return a synthetic dataset, where $N$ is the number of tuples in the original dataset.
As a baseline, we also compare with the \algoname{Identity} algorithm, which measures each element of the domain using the Laplace mechanism (\planref{\planidentity} in Fig. \ref{fig:operator_index})}

\setlength{\tabcolsep}{12pt}

\def\num#1{\numx#1}\def\numx#1e#2{{#1}\mathrm{e}{#2}}
\makeatletter
\newcolumntype{B}[3]{>{\boldmath\DC@{#1}{#2}{#3}}c<{\DC@end}}

\begin{table}
\small
    \caption{\label{table:privbayes} Results on Census data; domain size 1,400,000; scale of error is indicated under each workload.}

  \centering

  \begin{tabularx}{0.8\columnwidth}{l d{-3} d{-3} d{-3}}
&  \multicolumn{3}{c}{\bf Workload}  \\
 \multicolumn{1}{c}{\textbf{Algorithm}}       &
 \multicolumn{1}{c}{\shortstack{Identity \\  {\small $(\num{1e-9})$} }} &
 \multicolumn{1}{c}{\shortstack{2-way Marg. \\ {\small $(\num{1e-7})$} }} &
 \multicolumn{1}{c}{\shortstack{ Prefix (Income)  \\ {\small $(\num{1e-7})$} }} \\ \hline

   \algoname{Identity}        & 24.18   & 12.04   & 18.97     \\ \hline
   \algoname{PrivBayes}   & 76.93   & 65.31 & 28.70   \\ \hline
   \algoname{PrivbayesLS}   & 5.86  & 13.29   & 36.81     \\ \hline
   \algoname{Hb-Striped}        & 70.31   & 21.91   & 4.13  \\ \hline
   \algoname{Dawa-Striped}   & \multicolumn{1}{B{.}{.}{-1}}{3.43}   & \multicolumn{1}{B{.}{.}{-1}}{1.96}  & \multicolumn{1}{B{.}{.}{-1}}{2.50}    \\ \hline
  \end{tabularx}

\end{table}

Table \ref{table:privbayes} presents the results for each workload. We use scaled, per-query L2 error to measure accuracy.  First, we find that  \algoname{PrivBayes} performs worse than \algoname{Identity} on all workloads.  Interestingly, on Identity and 2-way marginal workloads, it is improved by our new plan  \algoname{PrivBayesLS} that replaces its inference step with least squares.  \algoname{PrivBayes} may be more suitable to input data with higher correlations between the attributes.
%
Second, our striped plans \algoname{Hb-Striped} and \algoname{Dawa-Striped} offer significant improvements in error. \algoname{Dawa-Striped} is the best performer: the data-dependent nature of DAWA exploits uniform regions in the partitioned data vectors.  This shows the benefit from \sys in allowing algorithm idioms designed for lower-dimensional data to be adapted to high dimensional problems.


\subsubsection{Naive Bayes classification}\label{sec:nb:exp}

\begin{figure}[h]
\centering
\includegraphics[width=.5\columnwidth]{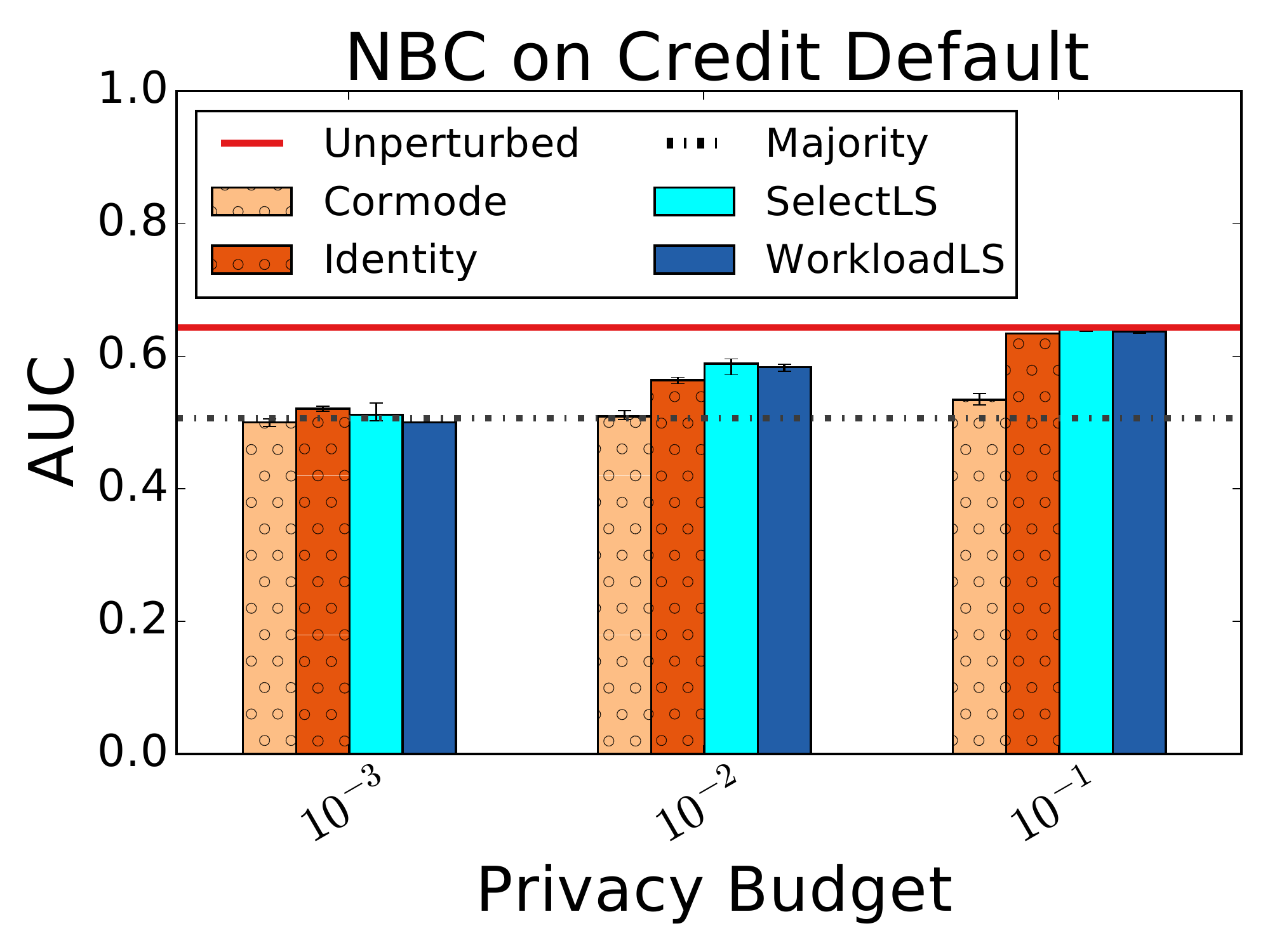}
\vspace{-2ex}
\caption{\label{fig:NBC} New \sys plans \algoname{WorkloadLS} and \algoname{SelectLS} result in NaiveBayes classifiers with lower error than plans that correspond to algorithms from prior work, and approach the accuracy of a non-private classifier for various $\epsilon$ values.}
\end{figure}

We evaluate the performance of the Naive Bayes classifier on \textit{Credit Default} \cite{Yeh20092473}, a  credit card clients dataset which we use to predict whether a client will default on their payment or not. The data consists of $30$k tuples and $24$ attributes from which one is the target binary variable ``Default'' and the rest are the predictive variables. We used the predictive variables $X_3-X_6$ for a total combined domain size of $17,248$.

In our experiments we measure the average area under the curve (AUC) of the receiver operating characteristic curve across a $10$-fold cross validation test. The AUC measures the probability that a randomly chosen positive instance will be ranked higher than a randomly chosen negative instance. We repeat  this process $10$ times  (for a total of $100$ unique testing/training splits) to account for the randomness of the differentially private algorithms and report the $\{25, 50, 75\}$-percentiles of the average AUC.
As a baseline we show the majority classifier, which always predicts the majority class of the training data and also show the unperturbed classifier as an upper bound for the utility of our algorithms.

In \cref{fig:NBC} we report our findings: each group of bars corresponds to a different $\epsilon$ value and each bar shows the median value of the AUC for an algorithm. For each DP algorithm we also plot the error bars at the $25$ and $75$ percentiles. The dotted line is plotted at $0.5067$ and shows the AUC of the majority classifier. The continuous red line is the performance of the non-private classifier (Unperturbed). For larger $\epsilon$ values we see that our plans significantly outperform the baseline and reach AUC levels close to the unperturbed. As $\epsilon$ decreases, the quality of the private classifiers degrades and for $\epsilon = 10^{-3}$ the noise added to the empirical distributions drowns the signal and the AUC of the private classifiers reach $0.5$, which is the performance of a random classifier. Our plan \algoname{WorkloadLS} is essentially the algorithm of \cite{Cormode11Differentially} with an extra inference operator.  This shows that the addition of an extra operator to a previous solution significantly increases its performance.

\subsection{Implementation comparison}
As discussed in \cref{sec:implementation}, most \sys operators involve performing operations on matrices, which can be implemented using several different representations: dense, sparse, and implicit.  All of them are lossless representations of the underlying matrix, so the choice of implementation does not influence plan accuracy. However, it could impact the efficiency and scalability.

In this section, we compare these alternative implementations.  We first evaluate how the choice of matrix implemenation impacts scalability and efficiency of plans.  Then we have a focused experiment on a key matrix operation: inference.

\subsubsection{Scalability and efficiency of plans}
\label{sec:sub:plan_efficiency}

\begin{figure}
  \centering
  \begin{subfigure}{\textwidth}
    \includegraphics[width=\textwidth]{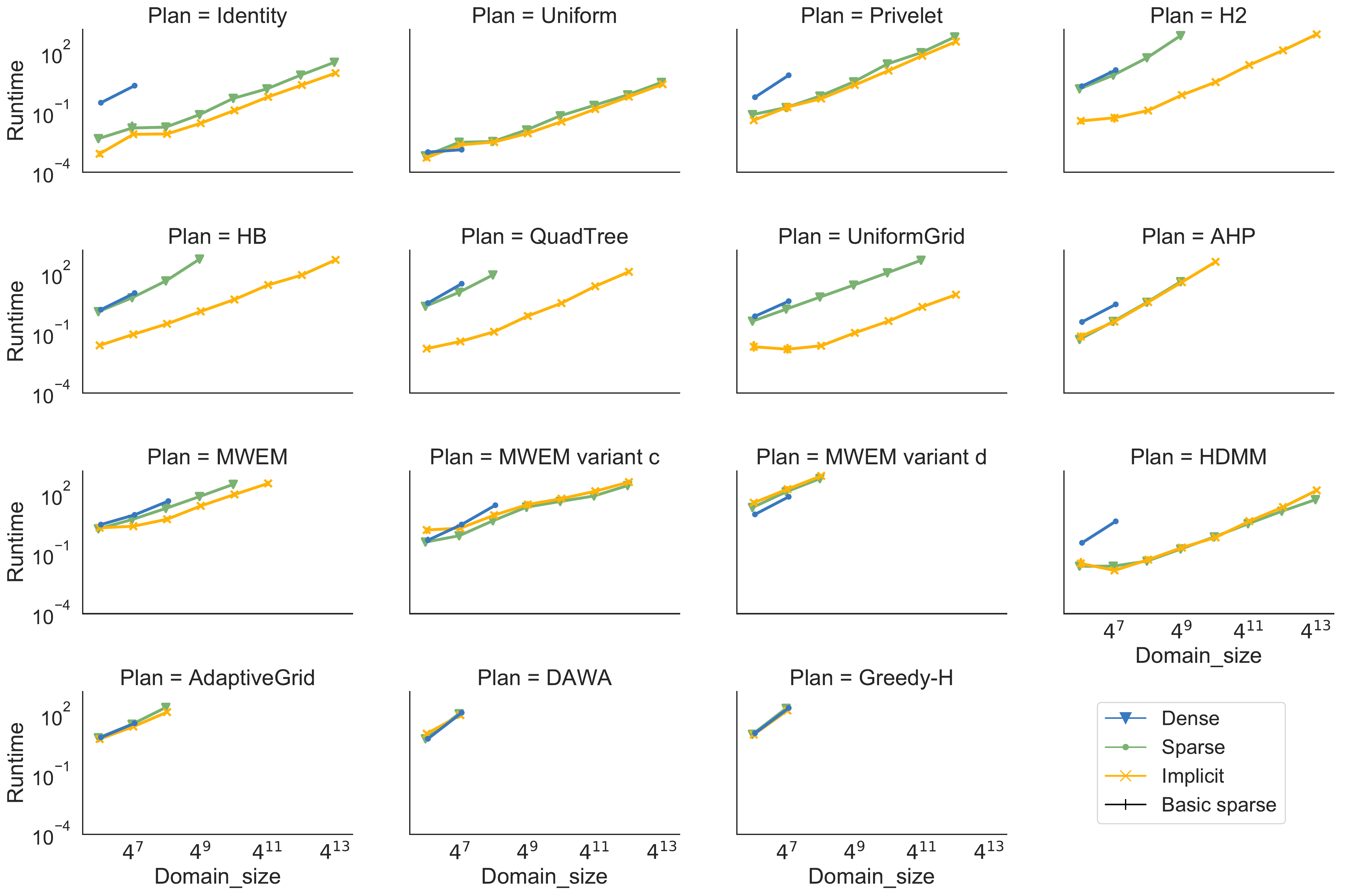}
    \caption{1-D and 2-D plans}
    \label{fig:plan_efficiency_ld}
  \end{subfigure}
  \begin{subfigure}{\textwidth}
    \includegraphics[width=\textwidth]{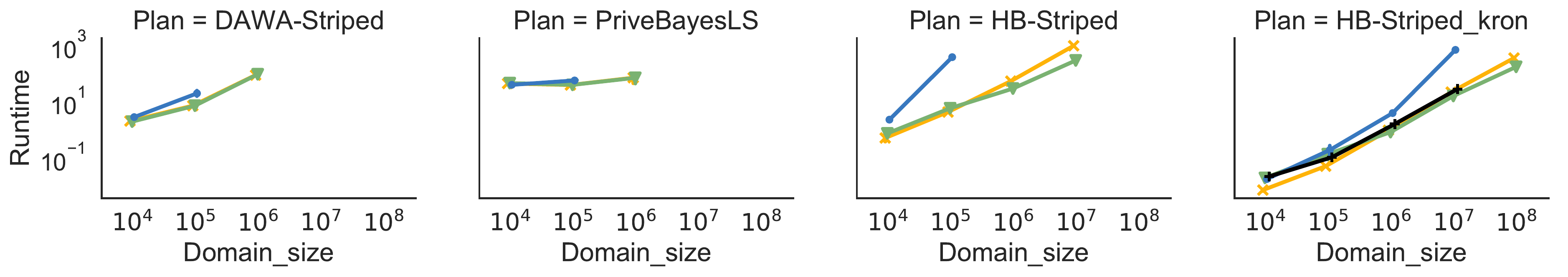}
    \caption{Multi-dimensional plans}
    \label{fig:plan_efficiency_hd}
  \end{subfigure}
  \caption{Plan execution time with different implementation of measurement matrices with Identity workload} 
  \label{fig:plan_efficiency}

\end{figure}

To understand the impact of different physical implementations for plans in \cref{fig:plan_signatures}, we compare the runtime of plans using implicit measurement matrices, which are new, with previous implementations from \cite{Zhang:2018:EFD:3183713.3196921}, which can use either dense or sparse matrices. For the plans \op{HDMM} and \op{HB-striped\_kron}, which contain operators that were not supported in the previous implementation, we make the comparison by converting implicit matrices to their sparse and dense representations. We measure the average end-to-end execution time over 5 random trials for three implementations of each plan along increasing domain sizes. We stop any execution when it runs for more than 1000s.

\cref{fig:plan_efficiency_ld} shows runtime for low-dimensional plans.  All plans are applied on two-dimensional domain square domains, except for \op{DAWA} and \op{Greedy-H}, which are designed for one-dimensional domains.  \planref \planmwemc\ \op{MWEM variant b} is omitted because it timed out even at the smallest domain we tested here.
The results show that for most plans, the implicit implementation has the best scalability. Also, looking at a fixed domain size, the implicit representation usually leads to faster runtime than its dense and sparse counterparts.
%

The performance improvement is most pronounced with plans \op{HB}, \op{QuadTree} and \op{UniformGrid} where implicit representation can scale to domains larger by a factor of 1000x. These algorithms construct hierarchical/grid-based measurement matrices and can be represented as {\em Range Queries}, a special instance of the implicit matrices. As discussed in \cref
{sec:range_queries}, this representation is compact and supports faster matrix-vector products.
There are few cases where the difference between implementations is less pronounced. \op{DAWA} and \op{Greedy-H} share the same special selection subroutine which needs to materialize the matrix. \op{AdaptiveGrid} has a plan that requires iterating through a potentially large number of partitions, and this step appears to dominate the runtime.

Results for high-dimensional plans are shown in \cref{fig:plan_efficiency_hd}.
%
%
For the first three plans, sparse and implicit representations exhibit similar performance and scale to domains at least 10x larger than using dense.
For the last plan, \op{HB-Striped\_kron}, recall from \cref{sec:census_plans} that this plan is an alternative way of expressing the same algorithm as the \op{HB-Striped} plan, but instead of partitioning the data, it uses Kronecker products to express queries compactly in terms of submatrices.  By comparing adjacent figures, we can see the approach based on Kronecker products allows plans to scale to at least 10x larger domains across implementations of the submatrices.  As another comparison point that illustrates the benefits of Kronecker products, in the last figure, we include ``Basic sparse'', an alternative implementation of the \op{HB-Striped\_kron} plan where the query Kronecker product matrix is replaced with a materialized sparse matrix over the full domain.


\subsubsection{Scalability of inference}

Inference is one of the most compu\-tation-intensive operators in \sys\, especially for large domains resulting from multidimensional data. Next, we show the impact of implementation choices on the scalability of inference. \cref{fig:scaling_inference} shows the computation time for running our main inference operators (LS and NNLS) as a function of data vector size.

Recall that the methods described in  \cref{sec:improved_inference} provide efficiency improvements by using iterative solution strategies ({\em iterative} instead of {\em direct} in the figure) and exploiting sparsity in the measurement matrix ({\em sparse} or {\em implicit} as opposed to {\em dense} in the figure). For this experiment, we fix the measured query set to consist of binary hierarchical measurements \cite{hay2010boosting}.  \cref{fig:scaling_inference} shows that using sparse matrices and iterative methods allow inference to scale to data vectors consisting of millions of counts on a single machine in less than a minute. The use of implicit matrices permits additional scale-up for both \op{LeastSquares} and \op{NNLS}.  We also compare against the inference method introduced by Hay et al., denoted `Tree-based' in the figure.  It  is an algorithm that is logically equivalent to \op{LeastSquares} but specialized for hierarchically structured measurements.  The general-purpose \op{LeastSquares} implementation is able to scale to much larger domains.

\vspace{-1ex}
\begin{figure}[h]
\centering
\includegraphics[width=0.9\columnwidth]{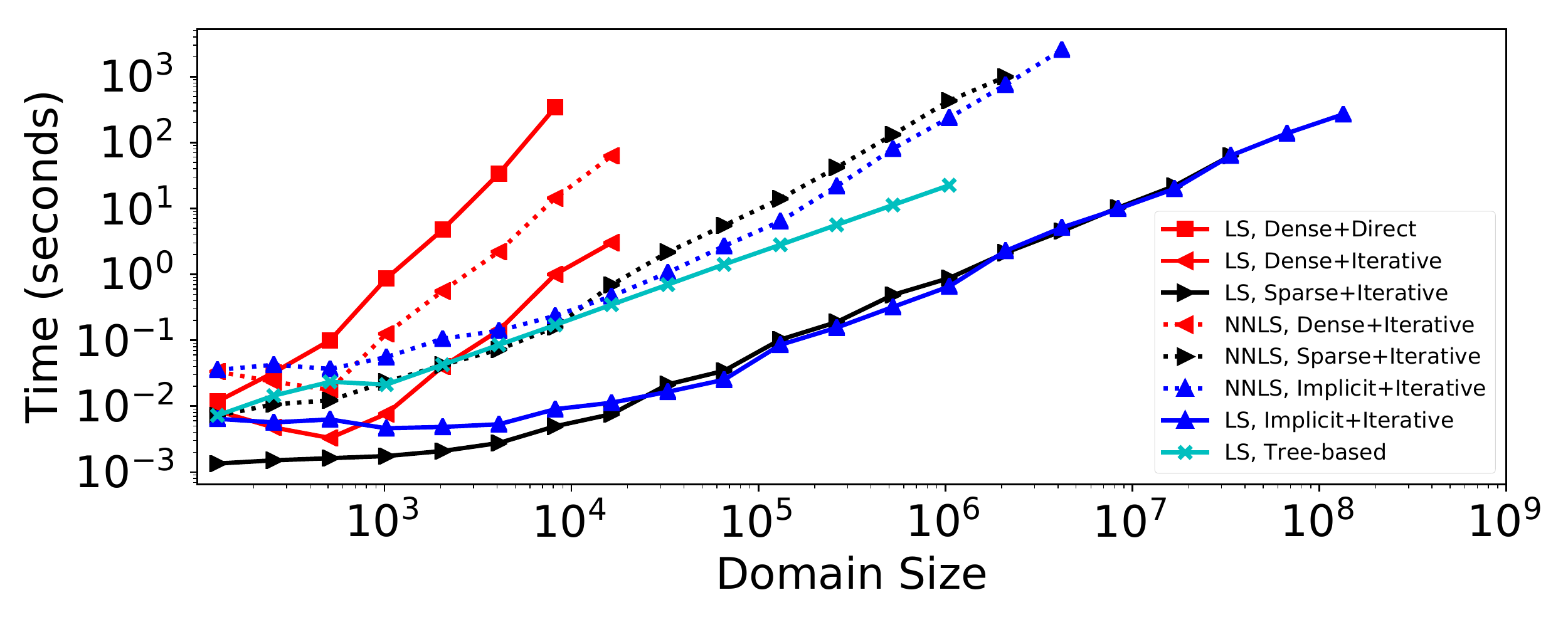}
\vspace{-2ex}
\caption{\label{fig:scaling_inference} For a given computation time, the proposed iterative and implicit inference methods permit scaling to data vector sizes as much as $1000\times$ larger than previous techniques using direct approaches and dense matrices.}
\end{figure}

\subsection{Workload-driven data reduction}
\label{sec:sub:wbdr_exp}
Next, we evaluate the impact of workload-driven data reduction, as described in Section~\ref{sec:sub:wkld}.  For selected algorithms, Table \ref{table:llreduce} shows that performing workload-driven data reduction improves \emph{both} error and runtime, almost universally.
\begin{table}[h]
\small
    \caption{\label{table:llreduce}Runtime (sec) and error improvements resulting from workload-based domain reduction. (W=RandomRange, small ranges. Original domain size: AHP (128,128), DAWA 4096, Identity (256,256), HB 4096)}

  \centering

  \begin{tabularx}{\columnwidth}{l  c  r | c  r | c  c }
 \multirow{2}{*}{ \textbf{ Algorithm} \vspace{15mm}} &
 \multicolumn{2}{c}{\shortstack{Original   \\ Domain}}   &
 \multicolumn{2}{c}{\shortstack{Reduced     \\ Domain}}  &
 \multicolumn{2}{c}{\shortstack{Factor \\ Improved}}  \\ \cline{2-7}
 &
\multicolumn{2}{c}{Error/Runtime}  &
\multicolumn{2}{c}{Error/Runtime}  &
\multicolumn{2}{c}{Error/Runtime} \\ \cline{1-7}

   \algoname{AHP}    & 1.68$\num{e-5}$  & 777.10  & 1.30$\num{e-5}$ & 145.00  & 1.29  & 5.36      \\ \cline{1-7}
   \algoname{DAWA}     & 1.06$\num{e-5}$  & 0.23  & 1.07$\num{e-5}$ & 0.25    & 0.99    & 0.92    \\ \cline{1-7}
   \algoname{Identity}& 4.74$\num{e-5}$ & 0.66  & 1.64$\num{e-5}$ & 0.90  & 2.89    & 0.73      \\ \cline{1-7}
   \algoname{HB}      & 3.20$\num{e-5}$ &0.05   & 2.38$\num{e-5}$ & 0.08    & 1.34  & 0.62      \\ \cline{1-7}

\end{tabularx}
\end{table}

The biggest improvement in error (a factor of 2.89) is witnessed for the \algoname{Identity} algorithm.  Without workload-driven reduction, groups of elements of the domain are estimated independently even though the workload only uses the total of the group.  After reduction, the sum of the group of elements is estimated and will have lower variance than the sum of independent measurements.

The biggest improvement in runtime occurs for the AHP algorithm.  This algorithm has an expensive clustering step, performed on each element of the data vector.  Workload-driven reduction reduces the cost of this step, since it is performed on a smaller data vector.  It also tends to improve error because higher-quality clusters are found on the reduced data representation.


\subsection{Summary of Findings}

The experiments evaluate the accuracy, scalability, and efficiency of \sys.
The case studies show that \sys can lead to more accurate algorithms with relatively little effort from the programmer: the MWEM algorithm can be improved significantly by replacing a few key operators; for the Census and Naive Bayes case studies, \sys can be used to design novel algorithms from existing building blocks, offering state-of-the-art error rates.
The study of scalability and efficiency found that implicit matrix representation can lead to huge performance gains, increasing scalability by a factor 1000x in some cases.  The generalized inference implementation scales well and outperforms specialized algorithms.
Finally, the evaluation shows the workload-driven data reduction improves accuracy and runtime, almost universally, so that it can be added to all workload-based plans with little cost and significant potential for gains.

\section{Related work} \label{sec:related}

\sys was first described by the authors in \cite{Zhang:2018:EFD:3183713.3196921}. The open-source codebase for \sys is publicly available and the results of \cite{Zhang:2018:EFD:3183713.3196921} are currently under SIGMOD reproducibility review. This manuscript extends \cite{Zhang:2018:EFD:3183713.3196921} by providing a unified and  improved approach to the efficient representation of matrix objects.  The new matrix representations are fundamental to \sys plans as they are used to represent workload queries, measurement queries, and partitions.  This innovation impacts many aspects of \sys, including inference and workload-based partition selection.  Existing plans were re-implemented using these new matrix techniques to measure their impact and were demonstrated to allow \sys plans to scale to far larger data vectors than previously possible.  

The implicit matrix representation furthers the notion that there can be a separation between the matrix as a logical representation and its physical implementation.  Implementation choices can be transparent to plan authors and open up interesting directions for developing highly optimized implementations.

A number of languages and programming frameworks have been proposed to make it easier for users to write private programs \cite{mcsherry2009pinq,Proserpio14Calibrating,Ebadi17Featherweight,Roy10Airavat:}.  The {\em Privacy Integrated Queries} (PINQ) platform began this line of work and is an important foundation for \sys.  We use the fundamentals of PINQ to ensure that plans implemented in \sys\ are differentially private.  In particular, we adapt and extend a formal model of a subset of PINQ features, called Featherweight PINQ~\cite{Ebadi17Featherweight}, to show that plans written using \sys operators satisfy differential privacy.  Our extension adds support for the partition operator, a valuable operator for designing complex plans.


Additionally, there is a growing literature on formal verification tools that prove that an algorithm satisfies differential privacy \cite{gaboardi:popl13,privinfer:ccs16,zhang:popl17}. For instance, LightDP \cite{zhang:popl17} is a simple imperative language in which differentially private programs can be written. LightDP allows for verification of sophisticated differentially private algorithms with little manual effort. LightDP's goal is orthogonal to that of \sys: it simplifies proofs of privacy, while \sys's goal is to simplify the design of algorithms that achieve high accuracy.
Nevertheless, an interesting future direction would be to implement \sys operators in LightDP to simplify both problems of verifying privacy and achieving high utility.




Concurrently with \cite{Zhang:2018:EFD:3183713.3196921}, Kellaris et al.~\cite{kellaris2015differentially} observed that algorithms for single-dimensional histogram tasks share subroutines that perform common functions.  The authors compare a number of existing algorithms along with new variants formed by combining subroutines, empirically evaluating trade-offs between accuracy and efficiency.  The focus is exclusively one-dimensional tasks.

The use of inference in differentially private algorithm design is not new~\cite{Williams2010,hay2010boosting,barak2007privacy}, and is used in various guises throughout recent work~\cite{xiao2010differential,Li14Data-,li2015matrix,zhangtowards,Acs2012compression,Cormode11Differentially,proserpio2012workflow,Lee15Maximum,mckenna2018optimizing}.
Proserpio et al.~\cite{proserpio2012workflow} propose a general-purpose inference engine based on MCMC that leverages properties of its operators to offset the otherwise high time/space cost of this form of inference.  Our work is complementary in that we focus on a different kind of inference (based on least squares) in part because it is used, often implicitly, in many published techniques.  A deeper investigation of alternative inference strategies is a compelling research direction.

Our use of implicit matrices was inspired by their use in \cite{mckenna2018optimizing}, where Kronecker products were used extensively to represent high-dimensional workloads and measurements.  \sys implicit matrices extend and generalize those matrix constructions. Techniques that use measurements based on wavelets \cite{barak2007privacy} (for range query workloads) and Fourier basis queries \cite{xiao2010differential} (for marginals), are examples in which measurement and inference is performed without materialization of a matrix, and so they could be seen as implicit methods. Hierarchical query sets \cite{hay2010boosting,Qardaji13Understanding} also admit inference methods that do not require materialization of a matrix.  Each of these examples relies on the special structure of the query sets they use in order to achieve scalability.  Our approach in \sys is more flexible: it allows the plan author to focus on \emph{what} to measure, rather than how to measure it and how to perform inference efficiently.  

\sys provides an execution framework for privacy algorithms, but does not perform plan-level optimization.  The matrix mechanism~\cite{li2015matrix} (and more recently the high-dimensional matrix mechanism~\cite{mckenna2018optimizing}) formulates an optimization problem that corresponds to query selection in \sys. The mechanism then estimates the selected queries and applies least squares inference.  This can be seen as a kind of optimization, but in a limited plan space which admits only data-independent plans.  Recent work~\cite{kotsogiannis2017pythia} examines the problem of algorithm selection---selecting the best algorithm for a given private dataset and task---and proposes a meta-algorithm, Pythia, capable of choosing among a set of ``black box'' algorithms. In contrast, \sys takes a ``white box'' approach, decomposing existing algorithms into modular operators and allowing plan authors to design new algorithms.  Pythia could be adapted to automatically select operators in \sys and, in fact, Pythia could be implemented as an \sys plan.

As noted above, our efficiency and scalability efforts have so far been focused on a centralized setting where \sys plans are executed on a single machine and data vectors fit in memory.  As such, \sys currently makes use Python's pandas and SciPy modules for relational and matrix processing, respectively.  In the future it may be beneficial for \sys to use a distributed data processing platform such as Apache's Hive\cite{hive} , Accumulo\cite{accumulo} , or Spark\cite{spark} , or to consider the innovations of academic projects such as Weld~\cite{shoumik:2017}, LaraDB~\cite{Hutchison:2017}, SPOOF~\cite{elgamal:2017}, and Samsara~\cite{schelter:2016}, many of which support parallelized matrix operations.  

However, these platforms do not provide easy scalability solutions for \sys, and are not substitutes for exploiting the special matrix structure present in \sys operators, which will be necessary and beneficial in any execution platform that is adopted.  (For example,  no computing platform is so efficient that it obviates the need to reduce a matrix from 56GB to less than 100 bytes, as in \cref{ex:implicit}.) Given our implicit matrix representations, the main bottleneck becomes the (always dense) representation of the data vector.  This could be distributed across a cluster and operated on using a variety of systems, however, there is reason to believe, due to limits of differentially private estimation, that most measurements will be performed on projections of the input relation, and will therefore result in a collection of in-memory data vectors rather than one monolithic vectorized dataset.

\section{Conclusions} \label{sec:conc}
\vspace{-1mm}
We have described the design and implementation of \sys: an extensible programming framework and system for defining and executing differentially private algorithms.  Many state-of-the-art differentially private algorithms can be specified as plans consisting of sequences of operators, increasing code reuse and facilitating more transparent algorithm comparisons.
Algorithms implemented in \sys are often faster and scale to larger domains by leveraging \sys's compact internal data representations, based on implicit matrices. 
Using \sys, we designed new algorithms that outperform the state of the art in accuracy on linear query answering tasks.

By allowing plan authors to focus on the simpler problem of designing a plan, and shifting the burden of implementing privacy and accuracy-critical operators to privacy engineers, we hope \sys will be a key driver in the wider adoption of differential privacy.

\sys is extensible and, through the addition of new operators, we hope to continue to expand the classes of tasks that can be supported.  For example, we would like to use \sys to build a differentially-private SQL query-answering system.  This requires a number of extensions, including support for specifying and enforcing more complex privacy policies over multiple input relations, extended relational transformations, and accompanying stability analysis.  In addition, we believe even greater scalability could be achieved by the addition of new inference operators that do not require full vectorization of the input data.  Lastly, \sys allows for a wide range of plans to be expressed and implemented in multiple ways, but it lacks plan-level automated optimization (some operators like HDMM and Greedy-H do perform operator-level optimization).  An optimizer for \sys would need to balance efficiency and accuracy metrics and reason about the equivalence of plans that contain randomized components.

{
\bibliographystyle{ACM-Reference-Format}
\bibliography{ms}
}



\end{document}